\documentclass[accepted, breaklinks,colorlinks=true,linkcolor=BrickRed, urlcolor=blue, anchorcolor=Blue,citecolor=Blue,backref=page]{uai2023} %
\usepackage[american]{babel}
\usepackage[authoryear,round]{natbib} %
\usepackage{mathtools} %
\usepackage{amssymb}
\usepackage{alphalph}
\usepackage{multirow}
\usepackage[table,dvipsnames]{xcolor}
\usepackage{amsthm}
\usepackage{booktabs}
\usepackage{bm}
\usepackage{bbm}
\usepackage{booktabs} %
\usepackage{tikz} %
\usepackage{notation}
\usepackage{hyperref}
\usepackage{subcaption}
\usepackage{cleveref}
\crefname{figure}{Fig.}{Figs.}
\crefname{definition}{Def.}{Defs.}
\crefname{corollary}{Cor.}{Cors.}
\crefname{proposition}{Prop.}{Props.}
\crefname{theorem}{Thm.}{Thms.}
\crefname{remark}{Remark}{Remarks}
\crefname{principle}{Principle}{Principles}
\crefname{lemma}{Lemma}{Lemmata}
\crefname{claim}{Claim}{Claims}
\crefname{table}{Tab.}{Tabs.}
\crefname{section}{\S}{\S\S}
\crefname{subsection}{\S}{\S\S}
\crefname{subsubsection}{\S}{\S\S}
\crefname{assumption}{Assumption}{Assumptions}
\crefname{appendix}{App.}{App.}
\crefname{equation}{}{}
\crefname{example}{Example}{Examples}
\newcommand{\ulim}[1]{\underset{#1 \rightarrow \infty}{\text{lim}}} %
\newcommand{\cons}{\xrightarrow{\text{P}}} %
\newcommand{\asconv}{\xrightarrow{\text{a.s.}}} %
\newcommand{\sampleiid}{\stackrel{\text{i.i.d.}}{\sim}} %

\newtheorem{theorem}{Theorem}[section]
\newtheorem{definition}{Definition}[section]
\newtheorem{proposition}[theorem]{Proposition}

\newtheorem{remark}[theorem]{Remark}
\renewcommand*{\backref}[1]{}
\renewcommand*{\backrefalt}[4]{%
    \ifcase #1%
          \or [Cited on page~#2.]%
          \else [Cited on pages~#2.]%
    \fi%
    }
\title{Causal Effect Estimation from Observational and Interventional Data \\ Through Matrix Weighted Linear Estimators}
\author[1,2]{Klaus-Rudolf Kladny%
}
\author[2,3]{Julius von K\"ugelgen}
\author[1,2]{Bernhard Sch\"olkopf}
\author[2]{Michael Muehlebach}
\affil[1]{%
    Department of Computer Science\\
    ETH Z\"urich\\
    Switzerland
}
\affil[2]{%
    Max Planck Institute for Intelligent Systems,
    T\"ubingen, Germany
}
\affil[3]{%
    Department of Engineering, University of Cambridge, United Kingdom
    
    \texttt{\{kkladny, jvk, bs, michaelm\}@tue.mpg.de}
    }
\begin{document}
\maketitle

\begin{abstract}
We study causal effect estimation from a mixture of observational and interventional data in a confounded linear regression model with multivariate treatments. 
We show that the statistical efficiency in terms of expected squared error can be improved by combining estimators arising from both the observational and interventional setting. 
To this end, we derive methods based on matrix weighted linear estimators and prove that our methods are asymptotically unbiased in the infinite sample limit. This is an important improvement compared to the pooled estimator using the union of interventional and observational data, for which the bias only vanishes if the ratio of observational to interventional data tends to zero. 
Studies on synthetic data confirm our theoretical findings. In settings where confounding is substantial and the ratio of observational to interventional data is large, our estimators outperform a Stein-type estimator and various other baselines.
\end{abstract}
\section{Introduction}\label{sec:intro}
Estimating the causal effect of a treatment variable on an outcome of interest is a fundamental scientific problem that is central to disciplines such as econometrics, epidemiology, and social science~\citep{angrist2009mostly,morgan2014counterfactuals,imbens2015causal,hernan2020causal}.
A fundamental obstacle to this task is the possibility of hidden confounding: unobserved variables that influence both the treatment and the outcome may introduce additional 
associations between them~\citep{reichenbach1956direction}.
As a result, estimators purely based on observational (passively collected) data can be biased and typically do not recover the true causal effect.

This contrasts experimental studies such as randomized controlled trials~\citep[RCTs;][]{neyman1923application,fisher1936design}, where the treatment assignment mechanism is modified through an external intervention, thus breaking potential influences of confounders on the treatment.
For this reason, RCTs have become the gold standard for causal effect estimation.
However, obtaining such interventional data is difficult in practice because the necessary experiments are often infeasible, unethical, or very costly to perform.

\begin{figure}[tb] 
\centering
\begin{subfigure}{0.5\columnwidth}
\centering
\includegraphics[width=10em]{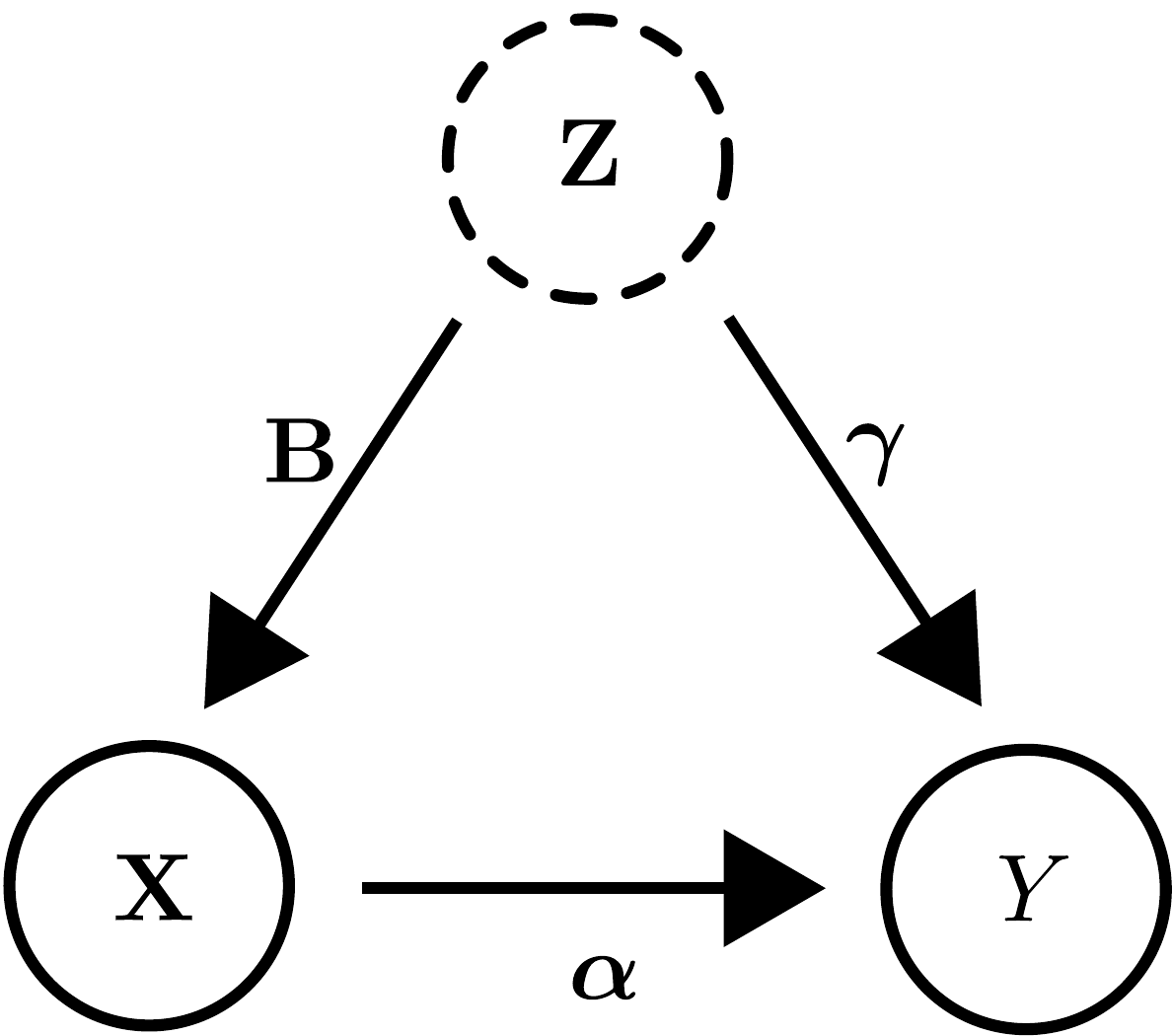}
\caption{observational}
\label{fig:setting_obs}
\end{subfigure}%
\begin{subfigure}{0.5\columnwidth}
\centering
\includegraphics[width=10em]{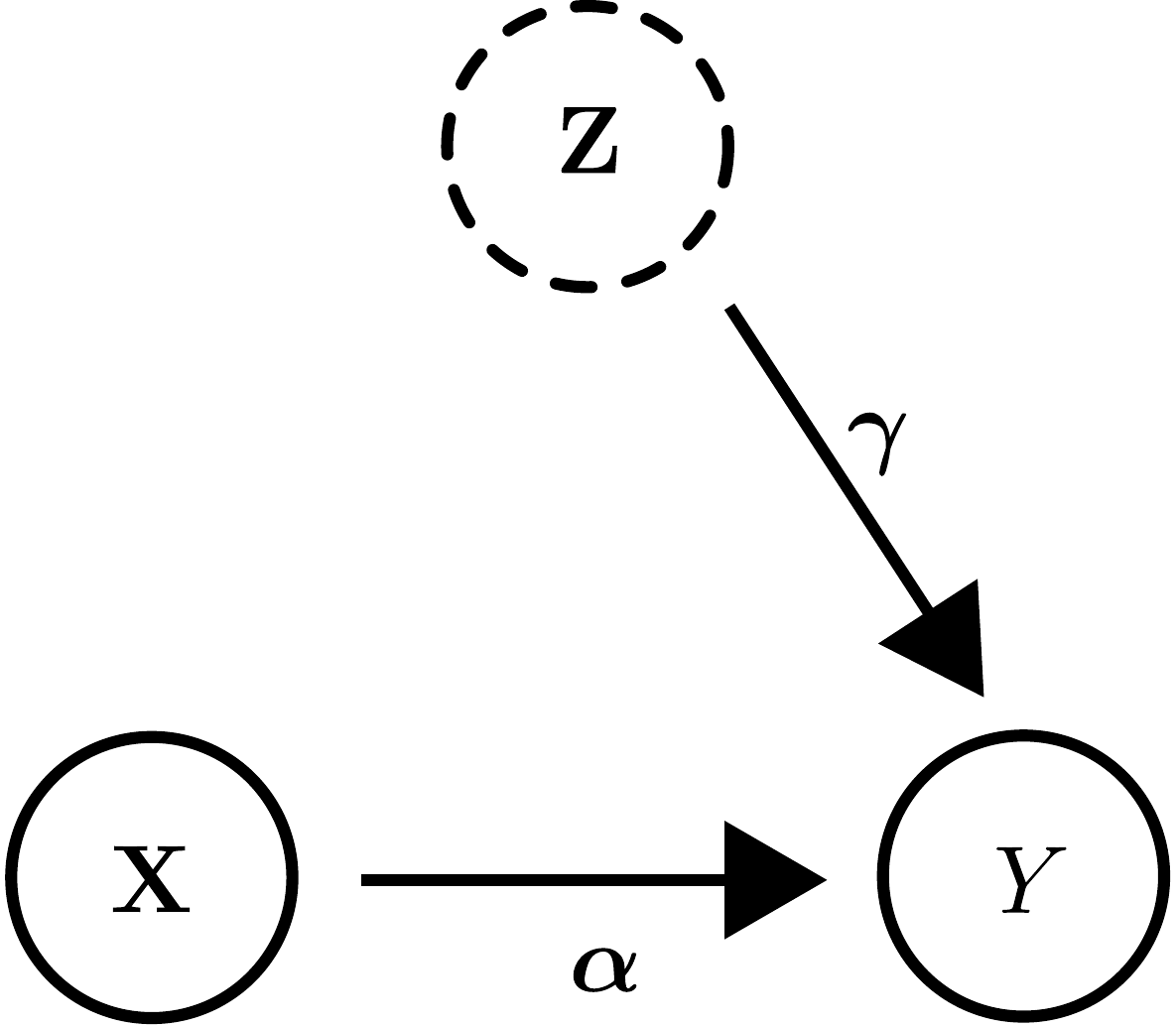}
\caption{interventional}
\label{fig:setting_int}
\end{subfigure}
\caption{\textbf{Overview.} We estimate the causal effect~$\alphab$ of a continuous multi-variate treatment~$\Xb$ on a scalar outcome~$Y$ in a linear Gaussian structural equation model with unobserved confounders~$\Zb$. 
Given a large sample from the observational setting in (a) and a small sample from the interventional setting in (b), we establish an optimal, sample-size dependent matrix weighting scheme for combining the biased, low-variance observational estimator with the unbiased, high-variance interventional estimator.}
\label{fig:scm}
\end{figure}

In contrast, observational data is usually cheap and abundant, motivating the study of causal inference from observational data~\citep{rubin1974estimating,pearl2009causality}.
In fact, in certain situations causal effects can be identified and estimated from purely observational data, even under hidden confounding, e.g., in the presence of natural experiments~\citep[instrumental variables;][]{angrist1996identification} or observed mediators~\citep[front-door adjustment;][]{pearl1995causal}.
However, this does not apply to the general case in which a treatment~$\Xb$ and an outcome~$Y$ are confounded by an unobserved variable~$\Zb$ as shown in~\cref{fig:setting_obs}.

\looseness-1 In the present work, we study treatment effect estimation in this general setting under the assumption that we have access to both observational and interventional data. The latter can be viewed as sampled from the setting shown in~\cref{fig:setting_int}, where the arrow $\Zb\to\Xb$ has been removed as a result of the intervention on $\Xb$~\citep[graph surgery;][]{spirtes2000pc}, and is thus unbiased for our task. 
Due to small sample size, however, the estimator based only on interventional data may exhibit high variance. 
Our main idea is therefore to use the (potentially large amounts of) observational data for variance reduction---at the cost of introducing some bias. This is achieved by forming a combined estimator, which is superior to the purely interventional one in terms of 
mean squared error. 

We make the key assumption that both the treatment $\Xb \to Y$ and confounding $\Zb\to\{\Xb,Y\}$ effects are linear, but allow for treatment~$\Xb$ and unobserved confounder~$\Zb$ to be continuous and multi-variate. 
We then consider a class of estimators of the causal effect parameter vector that combine the unbiased, but high-variance interventional estimator and the biased, but low-variance observational estimator through weight matrices---akin to a multi-variate convex combination.
We study the statistical properties of these estimators, establish theoretical optimality results, and investigate their empirical behavior through simulations.

In summary, we highlight the following contributions:
\begin{itemize}[topsep=0em,itemsep=0em,leftmargin=*]
    \item We introduce a new framework of weighing linear estimators using matrices and show that several existing approaches fall into this category~(\cref{sec:weighting_matrix}).
    \item \looseness-1 We prove that, unlike pooling observational and interventional data~(\cref{prop:MSE_greater_zero}), our matrix weighting approaches 
    achieve vanishing mean squared error in the interventional sample limit~(\cref{prop:weight_matrix_cons,theo:no_bias_in_the_limit}) if the ratio between observational and interventional data is non-vanishing.
    \item \looseness-1 We discuss two practical approaches for variance reduction in estimating optimal weight matrices (\cref{sec:inductive_biases}; \cref{prop:weak_weight_consistency}), and demonstrate through simulations that our estimators outperform 
    baselines and existing methods
    in situations where confounding is substantial~(\cref{sec:experiments}).
\end{itemize}
\section{Related Work}\label{sec:related_work}
Causal reasoning, i.e., inferring a causal query such as a causal effect, can be split up into the tasks of (i) identification and (ii) estimation.
Step (i) operates at the population level and seeks to answer whether a causal question can---at least in principle---be answered given infinite data.
If the answer is positive and a valid estimand is provided, step (ii) then aims to construct a statistically efficient estimator.

A causal query is identified from a set of assumptions if it can be expressed in terms of the available distributions (e.g., a mixture of different observational and interventional distributions). 
To this end, \citeauthor{pearl2009causality}'s do-calculus~(\citeyear{pearl1995causal,pearl2009causality}) provides an axiomatic set of rules for manipulating causal expressions based on graphical criteria.
The identification task has been studied extensively~\citep{tian2002general,pearl2014external,bareinboim2016causal} and has by now been solved for many settings of interest: In these cases, the do-calculus---and its extensions~\citep{correa2020general}---are sound and complete in that they provide a valid estimand if and only if one exists~\citep{huang2006identifiability,shpitser2006identification,bareinboim2012causal,sanghack2020gID}.

\looseness-1 In our setting from~\cref{fig:scm}, the causal effect~$\alphab$ is not identifiable from observational data, but is trivially identified by intervening on $\Xb$.
Yet, this leaves open the question of \textit{how to estimate $\alphab$ from finite data in the best possible way}.
In contrast to the plethora of works on identification, there is much less
prior literature about statistical efficiency of causal parameter estimation, particularly for confounded settings. 

\begin{figure}[tb]
    \centering
    \includegraphics[width=\columnwidth]{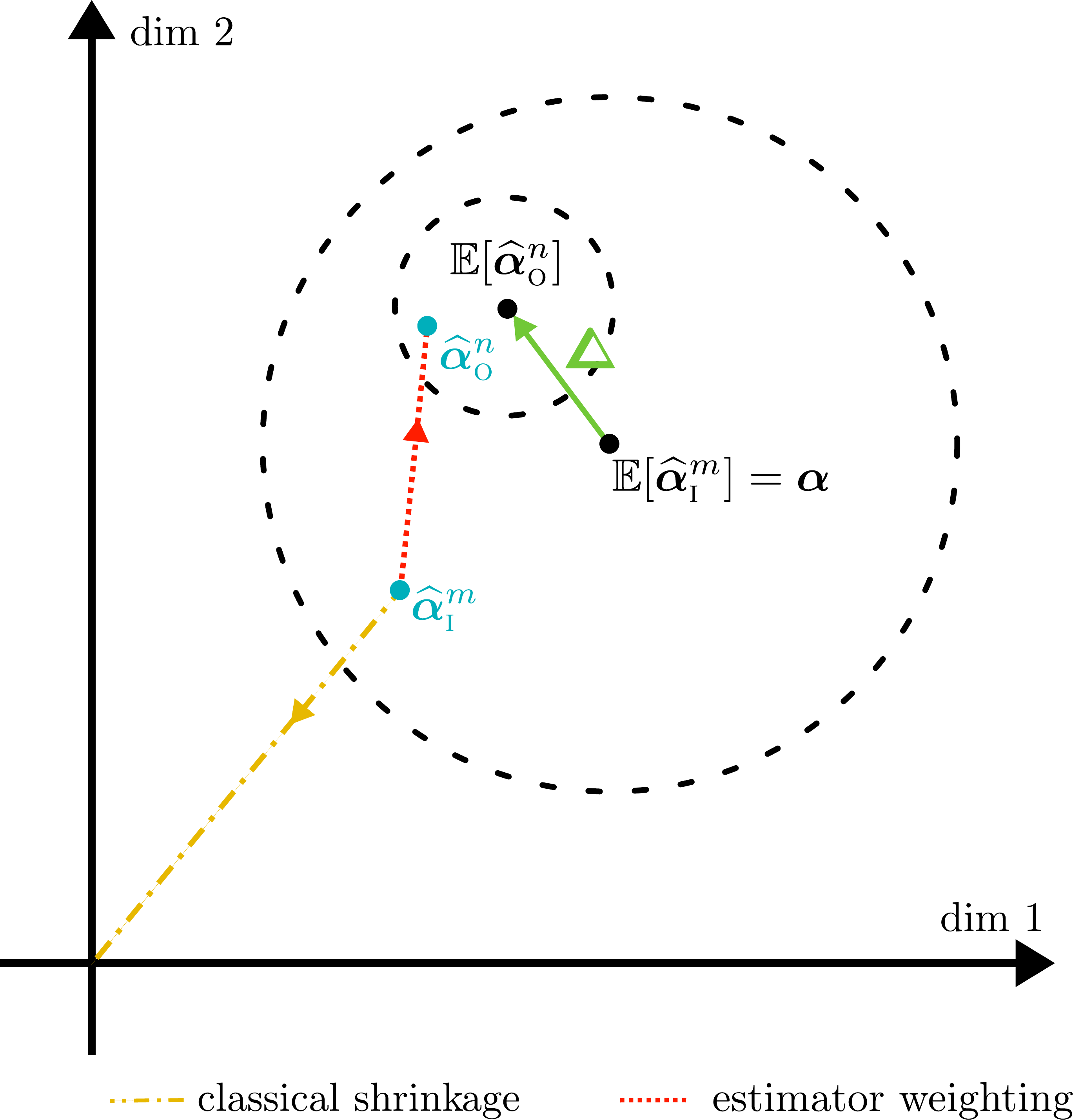}
    \caption{\textbf{Relation between shrinkage 
    and estimator weighting in 2D.} Classical shrinkage methods like ridge regression  (yellow) shrink the interventional estimator $\alphaI$ toward the origin. Scalar estimator weighting (red) instead shrinks toward the observational estimator~$\alphaO$. 
    Dashed circles show the covariances of $\alphaO$ and $\alphaI$, here assumed isotropic. $\bm\Delta$ (green) is the confounding-induced bias of $\alphaO$.
    }
    \label{fig:relation_to_shrink}
\end{figure}

A common source of inspiration for both prior work and our approach is that of shrinkage estimation. 
In light of the bias-variance decomposition of the mean squared error~\citep[e.g.,][p.~24]{hastie2009elements}, shrinkage can yield a strictly better (``dominating'') estimator by reducing variance, at the cost of introducing some bias.
These ideas were first introduced in frequentist statistics by~\citet{stein1956inadmissibility,james1961estimation} who showed that the maximum likelihood estimate of a multivariate mean is dominated by shrinking towards a fixed point such as the origin.
Similar ideas are also at the heart of empirical Bayes analysis~\citep{robbins1964empirical,efron1973stein,efron2012large}.
For estimating a parameter vector~$\alphab$ in a linear model, as is the focus of the present work, a classical shrinkage method is ridge regression~\citep{hoerl1970ridge}.
Instead of shrinking towards the origin, an intuitive idea for causal effect estimation
is to \textit{shrink towards the observational estimator}. The hope is that the latter constitutes a better attractor
if the confounding bias is not too large---despite  a slight increase in variance compared to shrinking toward a constant.
We refer to this approach as scalar estimator weighting. \Cref{fig:relation_to_shrink} shows a visual comparison to classical shrinkage estimation.
The most closely related work on estimator weighting is that of
\citet{strawderman1991stein, green2005improved} and \citet{rosenman2020combining}. The former two consider general biased and unbiased estimators. The latter propose weighting schemes for estimating vectors of multiple \textit{binary} treatment effects. These works are strongly inspired by James-Stein shrinkage estimators
and minimize a generalized version of Stein's unbiased risk estimate~\citep[p.~150]{wasserman2006allofnonparam}. \Citet{rosenman2020combining}
show optimality among scalar weights with respect to minimizing the true risk as the %
dimensionality of the estimated treatment effects goes to infinity. However, these theoretical results rely on knowledge of the true covariance matrix of the interventional estimator (which is typically unknown in practice), and 
the behavior of their estimators %
in the infinite sample limit is not analyzed.

Other work that focuses on combining observational and interventional data to estimate  causal effects of \textit{binary} treatments includes, e.g., \citet{kallus2018removing,cheng2021adaptive,ilse2021combining,rosenman2022propensity,hatt2022combining}, see~\citet{colnet2020causal} for a comprehensive survey.

\Citet{pmid33088006} also study combining estimators of binary treatment effects. However, in their framework an estimator with less bias in addition to a second error-prone estimator is computed from a second observational ``validation set'', in which all confounders are measured. Our framework, in contrast, does not require measurements of the confounders.
In the present work, we consider a
general
linear regression setting with continuous (rather than binary) multi-variate treatments. 
To combine observational and interventional data, we introduce a new class of matrix (rather than scalar) weighted estimators, of which ridge regression and data pooling are special cases.
Instead of employing Stein's unbiased risk estimate, we develop and analyze estimates for the theoretically optimal weight matrix, without making assumptions about the covariance structure of estimators.

Most approaches to causal estimation, including the present work, assume that the causal structure among variables is known and takes the general form of the directed acyclic graph in~\cref{fig:scm}.
For prior work on leveraging observational and interventional data for causal discovery, or structure learning, see, e.g., \citet{wang2017permutation}.

\section{Setting \& Preliminaries}\label{sec:setup}
\paragraph{Notation.} Upper case $Y$ denotes a scalar random variable, lower-case $y$ a scalar, bold lower-case $\xb$ a vector, and bold upper-case $\Xb$ either a matrix or random vector. 
The spectral norm of a matrix $\Xb$ is denoted by $\norm{\Xb}_2$.

\paragraph{Causal Model.} 
To formalize our problem setting, 
we adopt the structural causal model framework of~\citet{pearl2009causality}. Specifically, we assume that the causal relationships between the $d$-dimensional confounder~$\Zb$, the $p$-dimensional treatment~$\Xb$, and the scalar outcome~$Y$ are captured by the following linear Gaussian structural equation model~(SEM):
\begin{align} \label{equ:multivarSCMconfounder}
    \mathbf{Z}  &\; \leftarrow \; \mathbf{N}_{\mathbf{Z}}, 
    &&\Nb_\Zb\sim\Ncal(\bm{\mu}_{\Nb_\Zb}, \Sigmab_{\Nb_\Zb})\\
    \Xb  &\; \leftarrow \;  \mathbf{B} \mathbf{Z}+ \mathbf{N}_{\Xb}, \label{equ:SA_X} 
    &&\Nb_{\Xb}\sim\Ncal(\bm{\mu}_{\Nb_\Xb}, \Sigmab_{\Nb_\Xb})\\ 
    Y  &\; \leftarrow \; \mathbf{Z}^{\top} \bm\gamma+ \Xb^{\top} \bm\alpha+N_Y,
    &&N_Y\sim\Ncal(\mu_{N_Y}, \sigma^2_{N_Y}) \label{eq:outcome}
\end{align}
with $\Bb\in\RR^{p\times d}$, $\gammab\in\RR^d$, $\alphab\in\RR^p$, and $(\Nb_\Zb, \Nb_\Xb, N_Y)$ mutually independent exogenous noise variables. 
The SEM in~\cref{equ:multivarSCMconfounder}--\eqref{eq:outcome} induces an observational distribution over $(\Zb,\Xb, Y)$ which is referred to as $\Pobs$, see~\cref{fig:setting_obs}.

To model the interventional setting, we consider a soft intervention~\citep{eberhardt2007interventions},
which randomizes the treatment~$\Xb$ by replacing the assignment in~\cref{equ:SA_X} with
\begin{equation} \label{equ:replaced_SA}
    \Xb \; \leftarrow \; \widetilde{\mathbf{N}}_{\Xb},\qquad \qquad \widetilde{\mathbf{N}}_{\Xb}\sim\PP_{\widetilde{\mathbf{N}}_{\Xb}}, \qquad 
\end{equation}
where $\widetilde{\mathbf{N}}_{\Xb}$ is mutually independent of $\Nb_\Zb$ and $N_Y$. We note that $\widetilde{\mathbf{N}}_{\Xb}$ may be non-Gaussian. The modified interventional SEM consisting of~\cref{equ:multivarSCMconfounder,equ:replaced_SA,eq:outcome} induces a different, interventional distribution over $(\Zb,\Xb,Y)$, which we refer to as $\Pint$, see~\cref{fig:setting_int}. 

For ease of notation and for the remainder of this work, we assume without loss of generality that all noise variables are zero-mean. Details on how to extend our method to non zero-mean noise variables are provided in App.~\ref{sec:mean_non_zero}.

\paragraph{Data.} We assume access to two separate datasets of observations of $(\Xb,Y)$ of size $n$ and $m$, each sampled independently from the observational and interventional distributions (i.i.d.), respectively:
\begin{align*}
    (\xb_i, y_i) \; &\sampleiid \; 
    \Pobs, \quad i=1,...,n,\\
    (\xb_i, y_i) \; &\sampleiid \; 
    \Pint,\quad i=n+1,...,n+m,
\end{align*}
where $\Pobs$ and $\Pint$ denote the distributions of $(\Xb, Y)$ in the observational and interventional settings, respectively.
We note that the confounder $\Zb$ remains unobserved.
We concatenate the observational sample in a treatment matrix $\XO=(\xb_1, ..., \xb_n)^\top\in\RR^{n\times p}$ and outcome vector $\YO=(y_1,...,y_n)^\top\in\RR^n$, and similarly with $\XI,\YI$ for the interventional sample.
Finally, we denote the pooled data by  $\XP=(\XO, \XI)\in\RR^{(n+m)\times p}$ and $\YP=(\YO,\YI)\in\RR^{n+m}$. 

\paragraph{Goal.} Our objective is to obtain an accurate estimate of the parameter vector $\bm{\alpha}$, which characterizes the linear causal effect of $\Xb$ on $Y$ in~\cref{eq:outcome}. Formally, it is given by
\begin{align*}
    \bm{\alpha} = \nabla_{\xb} \mathbb{E} [Y|\mathrm{do}(\Xb \leftarrow \xb)],
\end{align*}
where the $\mathrm{do}(\cdot)$ operator denotes a manipulation of the treatment assignment akin to~\cref{equ:replaced_SA}, and the expectation is taken with respect to the corresponding conditional distribution.

\paragraph{Confounding Issues.}In the general case with non-zero $\Bb$ and $\gammab$, the observational setting is confounded, meaning %
\begin{equation*}
    \Pobs(Y | \Xb = \xb) \neq 
    \PP(Y | \mathrm{do}(\Xb \leftarrow \xb))=\Pint(Y | \Xb = \xb),
\end{equation*}
which complicates the use of observational data.
Specifically, for our assumed model~\cref{equ:multivarSCMconfounder}--\eqref{eq:outcome} the conditional expectation of $Y$ under $\Pobs$ is given by
the following perturbed linear model~\citep{cevid2020spectral}:
\begin{equation} \label{eq:pert_model}
    \mathbb{E}_\text{obs}
    [Y | \Xb = \xb]=(\bm{\alpha} + \bm{\Delta})^{\top} \xb,
\end{equation}
where $\Deltab\in\RR^p$ denotes the \textit{confounding bias}, which 
is given explicitly in terms of the model parameters as
\begin{equation}\label{equ:Delta}
    \bm{\Delta} = (\Sigmab_{\mathbf{N}_{\Xb}} + \mathbf{B} \Sigmab_{\Nb_\mathbf{Z}} \mathbf{B}^{\top})^{-1} \mathbf{B} \Sigmab_{\Nb_\mathbf{Z}} \bm{\gamma}.
\end{equation}
It can be seen from~\cref{equ:Delta} that the confounding bias~$\Deltab$ is zero if $\Bb$
 or $\gammab$ are zero (i.e., $\Zb$ only affects either $\Xb$ or $Y$).
 Furthermore, we have that, in general,
\begin{equation}\label{eq:cond_vars}
      \text{Var}_{\text{obs}}(Y | \Xb)=\condvar \neq \intvar=\text{Var}_\text{int}(Y|\Xb)\,.
\end{equation}

\paragraph{Assessing Estimator Quality.}
We rely on
mean squared error with respect to the true parameter $\bm{\alpha}$ as a measure for comparing different estimators.
\begin{definition}[MSE]
Let $\widehat{\bm{\alpha}}$ be any function of the pooled data $(\XP,\YP)$ taking values in $\mathbb{R}^p$. Then
\begin{equation*}
    \MSE(\widehat{\bm{\alpha}}) \coloneqq \mathbb{E}\left[\norm{\widehat{\bm{\alpha}} - \bm{\alpha}}_2^2\right],
\end{equation*}
where the expectation is taken over $\XP, \YP$.
\end{definition}
We note that the mean squared error can also be written as follows:
\begin{equation}
\label{eq:bias_variance}
    \MSE(\widehat{\bm{\alpha}}) = \norm{\mathbf{Bias}(\widehat\alphab)}_2^2
    + \mathrm{Tr}(\mathbf{Cov}(\widehat{\bm{\alpha}}))\, ,
\end{equation}
where
\begin{equation*}
    \begin{aligned}
    \mathbf{Bias}(\widehat\alphab)&=\EE[\widehat\alphab] -\alphab\, ,\\
    \mathbf{Cov}(\widehat\alphab)&=\EE[(\widehat\alphab-\EE[\widehat\alphab])(\widehat\alphab-\EE[\widehat\alphab])^\top]\, .
    \end{aligned}
\end{equation*}
This decomposition highlights that biased estimators can dominate unbiased ones through variance reduction.

\paragraph{Pure Estimators.}
We study estimators for $\alphab$ that are linear combinations of the following ordinary least squares estimators  obtained on the two data sets individually.
\begin{definition}[Pure Estimators]
\label{def:pure_estimators}
For non-singular moment matrices $\XO^{\top} \XO$ and $\XI^{\top} \XI$,
the \emph{pure estimators} based only on the observational/interventional sample are given by:
\begin{align*}
    \alphaO &\coloneqq (\XO^{\top} \XO)^{-1} \XO^{\top} \YO,\\
    \alphaI &\coloneqq (\XI^{\top} \XI)^{-1} \XI^{\top} \YI.
\end{align*}
\end{definition}
Recall that $\alphaI$ is unbiased while $\alphaO$ has bias $\Deltab$. Their covariances conditionally on $\XO$ and $\XI$ are given by
\begin{equation}
\label{eq:estimator_covariances_true}
\begin{aligned}
    \COV(\alphaO)&=(\XO^{\top}\XO)^{-1}\condvar,
    \\
    \COV(\alphaI)&=(\XI^{\top}\XI)^{-1}\intvar.
\end{aligned}
\end{equation}
Unlike previous work~(see \cref{sec:related_work}), we do not make assumptions about the covariance structure of either estimator.

\paragraph{Almost sure convergence.} To analyze the behavior of estimators in the infinite sample limit, we will employ the following characterization known as \textit{almost sure convergence}.
\begin{definition}[Almost Sure Convergence]
\label{def:consistency}
Let $\mathbf{M}$ be a random matrix with realizations in $\mathbb{R}^{p \times p}$. We say a sequence of random matrices $\widehat{\mathbf{M}}_m$ indexed by $m \in \mathbb{N}$ \emph{converges almost surely} to ~$\mathbf{M}$, denoted $\widehat{\mathbf{M}}_m \asconv \mathbf{M}$, if and only if
\begin{equation*}
    \lim_{m\to\infty} \text{P} \left( \widehat{\mathbf{M}}_m  = \mathbf{M} \right) = 1,
\end{equation*}

where $\text{P}$ denotes probability.
\end{definition}

\section{Matrix Weighted Linear Estimators}\label{sec:weighting_matrix}

We now introduce our class of matrix weighted linear estimators, which combine the two pure estimators from~\cref{def:pure_estimators} using a weight matrix $\Wb$
to obtain a new (better) estimator.

\begin{definition}[$\Wb$-weighted Linear Estimator] \label{def:mat_weight_est}
Let $\Wb \in \mathbb{R}^{p \times p}$ (possibly random). The \emph{$\Wb$-weighted linear estimator} for $\alphab$ is given by
\begin{equation*}
    \alphaWeight{\Wb} \coloneqq \Wb \alphaI + (\mathbf{I}_p - \Wb)  \alphaO.
\end{equation*}
We furthermore refer to $\Wb$ as a weight matrix.
\end{definition}
We will generally think of $n$ as a function of $m$, where we sometimes even explicitly write $n(m)$. However, to simplify notation we index estimators by $m$ only, omitting the dependence $n(m)$.

Note that the purely interventional estimator is a special case of a $\Wb$-weighted estimator with $\Wb=\Ib_p$.
However, while unbiased, it may be subject to high variance if $m$ is very small.\footnote{E.g., consider a one-dimensional setting with $x_i = 1$ if $i$ is even 
and $-1$ otherwise. Then, for odd $m$, $\text{Var}(\alphaI|\XI) \propto (\sum_i x_i^2)^{-1}=\frac{1}{m}$.
} Hence, we generally prefer to employ the observational data as well and choose $\Wb\neq \Ib_p$.

\subsection{Existing Methods as Special Cases}

First, we show that several standard approaches can be viewed as special cases of matrix-weighted estimators.

\paragraph{Data Pooling.}  %
A straightforward approach for combining both data sets 
is to 
compute an estimator on the pooled data. 
The resulting least-squares estimator $\alphaP$ is: 
\begin{equation}
\begin{aligned} \label{equ:data_pooling}
    \alphaP :=&\, (\XP^{\top}\XP)^{-1} \XP^{\top} \YP \\ 
    =&\, (\XO^{\top}\XO + \XI^{\top} \XI)^{-1} (\XO^{\top} \YO + \XI^{\top} \YI) \\
    =&\, \WP  \alphaI + (\mathbf{I} - \WP)  \alphaO,
\end{aligned}
\end{equation}
where 
\begin{equation}\label{eq:W_data_pooling}
    \WP \coloneqq (\XO^{\top}\XO + \XI^{\top} \XI)^{-1} \XI^{\top} \XI.
\end{equation}
We see that
$\alphaP$ indeed qualifies as a valid matrix weighted estimator in the sense of~\cref{def:mat_weight_est}.

However, data pooling can lead to highly undesirable limiting behavior in cases where the amount of observational data~$n(m)$ does not vanish in the limit of infinite interventional data~$m \rightarrow \infty$. 
An example for this is given in the following proposition.
\begin{proposition} \label{prop:MSE_greater_zero}
Let $\lim_{m\to\infty} \frac{n(m)}{m} = c$ for some $c > 0$ and $\bm{\Delta} \neq \mathbf{0}$. Then, it holds that
\begin{equation*}
    \lim_{m\to\infty} \MSE \left(\alphaP \right) > 0.
\end{equation*}
\end{proposition}
The proof of~\cref{prop:MSE_greater_zero} is provided in App.~\ref{sec:proof_prop_MSE_greater_zero}. We note, however, that this does not happen for a vanishing amount of observational data, that is $\lim_{m\to\infty}\frac{n(m)}{m} = 0$ (see Prop.~4.2 in App.~\ref{sec:prop_pool_vanishing_MSE}.

\paragraph{Ridge Regression.} %
The ridge regression estimator on the interventional data, which shrinks~$\alphaI$ towards the origin (see~\cref{sec:related_work} and~\cref{fig:relation_to_shrink}),
is given by
\begin{align*}
    \widehat{\bm{\alpha}}^m_{\text{ridge}} &= (\XI^{\top}\XI + \lambda \mathbf{I}_p)^{-1} \XI^{\top} \YI \\
    &=  \Wridge \alphaI + (\mathbf{I}_p - \Wridge) \mathbf{0},
\end{align*}
where
\begin{equation}
\label{eq:W_ridge}
    \Wridge \coloneqq (\XI^{\top}\XI + \lambda \mathbf{I}_p)^{-1} \XI^{\top}\XI.
\end{equation}
Hence, $\widehat{\bm{\alpha}}^m_{\text{ridge}}$ can also be seen as a special case of a matrix weighted estimator with no observational data and $\alphaO=\mathbf{0}$.
Further, comparing~\cref{eq:W_data_pooling,eq:W_ridge} suggest an interpretation of
ridge regression 
as a \textit{poor man's} data pooling since access to observational data is replaced by a positive definite data matrix $\lambda \mathbf{I}_p$.
However, $\lambda$ is a constant, and therefore $\lim_{m\to\infty} \MSE (\widehat{\bm{\alpha}}^m_{\text{ridge}}) = 0$ even in the setting of ~\cref{prop:MSE_greater_zero}, which contrasts data pooling.

\subsection{Optimal Weighting Schemes} \label{sec:optimal_weighting_matrix}
We now establish theoretically optimal weighting schemes that minimize the mean squared error of $\Wb$-weighted linear estimators $\widehat\alphab^m_\Wb$ for different classes of weight matrices~$\Wb$ by exploiting the specific structure of our problem setting~(\cref{sec:setup}).

\paragraph{Optimal Scalar Weight.}
First, we consider the special case of scalar estimator weighting by considering weight matrices of the form $\Wb=w\Ib_p$
with weight $w \in [0, 1]$.
The optimal scalar weight is then derived as follows:
\begin{align*}
    &\frac{\partial}{\partial w} \MSE\left(\alphaWeight{w\Ib_p}\right) \overset{!}{=} 0
    \\
    \stackrel{\eqref{eq:bias_variance}}{\iff} &\frac{\partial}{\partial w} \left( \norm{\mathbb{E}[\alphaWeight{w\Ib_p} - \bm{\alpha}]}_2^2 + \text{Tr}\left(\mathbf{Cov}\left(\alphaWeight{w\Ib_p}\right)\right)\right) \overset{!}{=} 0
    \\
    \implies &\wopt = \frac{\text{Tr}(\mathbf{Cov}(\alphaO)) + \norm{\bm{\Delta}}_2^2}{\text{Tr}(\mathbf{Cov}(\alphaI)) + \text{Tr}(\mathbf{Cov}(\alphaO)) + \norm{\bm{\Delta}}_2^2}.
\end{align*}

\paragraph{Optimal Diagonal Weight Matrix.}
A more general case is to weigh each dimension individually by different scalars 
$w^{(k)} \in [0, 1], \; k = 1, ..., p$, corresponding to a weight matrix of the form $\Wb=
\mathrm{diag}(\wb)$.
The optimal diagonal weighting $\text{diag}(\wb^{m}_{*}$) is then given by
\begin{equation*}\label{equ:theor_optimal_diagonal}
    w_*^{m (k)} = \frac{\mathrm{Cov}^{(k, k)}(\alphaO) + \Delta^{(k) \, 2}}{\mathrm{Cov}^{(k, k)}(\alphaI) + \mathrm{Cov}^{(k, k)}(\alphaO) + \Delta^{(k) \, 2}}\,, 
\end{equation*}
for $k = 1, ..., p$.
The derivation is analogous to that for the optimal scalar weight above, with the only difference being that we optimize over each dimension separately.

\paragraph{Optimal Weight Matrix.}
Finally, we can also determine the optimum weighting as follows:
\begin{align}
\begin{split} \label{equ:theor_optimal_matrix}
    \Wopt = &\left( \COV(\alphaO) + \bm{\Delta} \bm{\Delta}^{\top} \right) \\ &\left( \COV(\alphaI) + \COV(\alphaO) + \bm{\Delta} \bm{\Delta}^{\top} \right)^{-1}.
\end{split}
\end{align}
A thorough derivation of the proposed weighting schemes can be found in App.~\ref{sec:detailed_derivation}. In addition, we elaborate on how this weighting scheme handles sample imbalance in App. ~\ref{sec:sample_imbalance}.

\begin{remark} \label{rem:recover_data_pooling}
If (i) $\bm{\Delta} = \mathbf{0}$ and (ii) $\condvar = \intvar$, then $\Wopt  = \WP$, i.e., data pooling corresponds to weighing with the optimal weight matrix under these two assumptions.
\end{remark}
\looseness-1 \Cref{rem:recover_data_pooling} can be verified by  simplifying~\cref{equ:theor_optimal_matrix} with assumptions (i) and (ii) and comparing to~\cref{eq:W_data_pooling}. It agrees with our intuition: Ordinary least squares relies on the assumption that $\mathbb{E}[Y | \Xb=\xb_i] = \alphab^\top \xb_i$ with equal variance, for all $i$. Thus, data pooling recovers the optimal estimator if these assumptions are true, i.e., the two conditional distributions $\Pobs(Y|X)$ and $\Pint(Y|\mathrm{do}(X))$ are identical. However, in general, they will not be identical and data pooling then amounts to model misspecification. This is likely to result in a non-vanishing mean squared error for $m \rightarrow \infty$ as highlighted in~\cref{prop:MSE_greater_zero}.

\subsection{Practical Estimators} \label{sec:practicality}
Unfortunately, the optimal weighting derived in~\cref{equ:theor_optimal_matrix} cannot be implemented directly, since the quantities 
$\bm{\Delta}$, $\COV(\alphaO)$, and $\COV(\alphaI)$ are unknown in practice.
To construct practical estimators informed by our theoretical insights, one option is thus to
rely on plug-in estimates of these unknown quantities. 
For $\COV(\alphaI)$ and $\COV(\alphaO)$, we use the standard estimators
\begin{equation*}
\begin{aligned}
    \widehat{\COV}(\alphaI) &= (\XI^{\top} \XI)^{-1} \widehat{\sigma}^2_{Y|\mathrm{do}(\Xb)},
    \\
    \widehat{\COV}(\alphaO) &= (\XO^{\top} \XO)^{-1} \widehat{\sigma}^2_{Y|\Xb},
\end{aligned}
\end{equation*}
which replace the conditional variances in~\cref{eq:estimator_covariances_true} by
\begin{equation*}
\begin{aligned}
    \widehat{\sigma}^2_{Y|\mathrm{do}(\Xb)} &= \frac{1}{m-1} \norm{ \YI - \XI \alphaI }_2^2,
    \\
    \widehat{\sigma}^2_{Y|\Xb} &= \frac{1}{n-1} \norm{ \YO - \XO \alphaO }_2^2.
\end{aligned}
\end{equation*}
For $\bm\Delta$, one may consider using the unbiased estimator
\begin{equation}
\label{eq:delta_plugin}
    \Deltaest = \alphaO - \alphaI.
\end{equation}
Substituting these into~\cref{equ:theor_optimal_matrix} then yields:
\begin{align} \label{eq:opt_matrix_estimator}
\begin{split}
    \Wopthat = &\left( \widehat{\COV}(\alphaO) + \Deltaest \Deltaest^{\top} + \epsilon \Ib_p \right) \\ &\left( \widehat{\COV}(\alphaI) + \widehat{\COV}(\alphaO) + \Deltaest \Deltaest^{\top} + \epsilon \Ib_p \right)^{-1}.
\end{split}
\end{align}
The regularization with $\epsilon > 0$ ensures that the inverse remains stable even in the large sample limit where $\widehat{\COV}(\alphaO)$ and $\widehat{\COV}(\alphaI)$ tend to zero. The reason for instability without such regularization is that $\Wopt$ is not uniquely defined in the infinite sample limit. With regularization, however, we can guarantee that $\Wopthat$  converges to $\Ib_p$ almost surely.

\begin{proposition}[Weight Matrix Convergence] \label{prop:weight_matrix_cons}
    Let $\lim_{m\to\infty} \frac{n(m)}{m}=c$, for some constant $c > 0$. Then, 
    $\Wopthat$ from~\cref{eq:opt_matrix_estimator} converges almost surely to $\Ib_p$, i.e., $\Wopthat \asconv \Ib_p$.
\end{proposition}
The proof for~\cref{prop:weight_matrix_cons} is included in App.~\ref{sec:proof_weight_matrix_cons}. We can show that this convergence
implies that the mean squared error vanishes asymptotically.
\begin{theorem}[Zero Mean Squared Error in the Sample Limit]\label{theo:no_bias_in_the_limit}
    Let $\West$ be any sequence of random weight matrices such that $\West \asconv \Ib_p$ and $\lim_{m\to\infty} \frac{n(m)}{m} = c$ for some constant $c>0$.
    Then,
    \begin{equation*}
        \ulim{m} \; \MSE \left( \alphaWeight{\West} \right) = 0,
    \end{equation*}
    where $\alphaWeight{\West}$ denotes the matrix-weighted linear estimator with weight matrix $\West$, as defined in~\cref{def:mat_weight_est}.
\end{theorem}
The proof of Thm. \ref{theo:no_bias_in_the_limit} is included in App.~\ref{sec:proof_no_bias_in_the_limit}. 

\Cref{theo:no_bias_in_the_limit} has the following relevant implication:  
we can incorporate an arbitrarily large amount of biased observational data and are still guaranteed that the bias (and also variance) of $\alphaWeight{\Wopthat}$ will vanish in the infinite sample limit. Moreover, 
this guarantee is independent of $\bm{\Delta}$ and $| \sigma_{Y|\Xb}^2 - \sigma_{Y|\mathrm{do}(\Xb)}^2 |$. 

We also note that \Cref{theo:no_bias_in_the_limit} does not imply unbiasedness of $\alphaWeight{\Wopthat}$ for any finite sample size. 

Further, we note
that almost sure convergence of $\West$ to $\Ib_p$ may generally not
be the only option to achieve vanishing mean squared error.
For example, if $\Deltab = \mathbf{0}$ such that $\alphaO$ is unbiased, we also obtain vanishing mean squared error for almost sure convergence of $\West$ to $\mathbf{0}$.

\subsection{Suitable Inductive Biases}\label{sec:inductive_biases}
Despite the desirable performance established in~\cref{theo:no_bias_in_the_limit}, the
plug-in estimates from~\cref{sec:practicality} will often not perform very well in finite sample settings. The main issue is the estimation of $\bm{\Delta}$, which has a large variance when done according to~\cref{eq:delta_plugin}.
To see this, we first note that
\begin{equation}
\label{eq:delta_covar_decomp}
    \text{Tr}(\COV(\Deltaest)) =  \text{Tr}(\COV(\alphaI)) + \text{Tr}(\COV(\alphaO)),
\end{equation}
since the observational and interventional data are independent.
Now, if we only have a small interventional sample (as is typically the case), 
$\text{Tr}(\COV(\alphaI))$ and hence according to~\cref{eq:delta_covar_decomp} also $\text{Tr}(\COV(\Deltaest))$ will be large.

We therefore explore possible 
inductive biases in the form of additional assumptions on the type of confounding that lead to reduced variance when estimating $\Deltaest$. These inductive biases can be motivated from domain knowledge and validation techniques such as cross-validation \citep{schaffer1993selecting}. Specifically, the application itself may provide some prior knowledge about the nature of confounding, which can then be confirmed by a better validation score compared to 
the other inductive biases/methods proposed here.

To this end, we observe that~\cref{eq:delta_plugin} can be written
as the solution of the following two-step ordinary least squares procedure:
\begin{align}
    &\alphaO  &&\gets \; \mathrm{arg} \; \min_{\bm{\alpha} \in \mathbb{R}^p} \; \left\{ \norm{ \YO \; - \; \XO \bm{\alpha} }_2^2 \right\} \nonumber \\
    &\mathbf{r}  &&\gets \; \YI - \XI \alphaO \nonumber \\
    &\Deltaest  &&\gets \; \mathrm{arg} \; \min_{\bm{\Delta} \in \mathbb{R}^p} \; \left\{ \norm{ \mathbf{r} \; + \; \XI \bm{\Delta} }_2^2 \right\}.\label{equ:optimization}
\end{align}

\paragraph{Small $\norm{\bm{\Delta}}_2$.}
In some settings, we may be willing to assume that, despite the existence of unobserved confounders, the resulting confounding bias is rather weak, i.e., that its Euclidean norm $\norm{\bm{\Delta}}_2$ is small. 
Since this is precisely the assumption underlying ridge regression, 
we reformulate~%
\eqref{equ:optimization} using a regularizer $\lambda_{\ell^2} > 0$ as
\begin{equation*}
    \Deltaest^{\ell^2} \; \gets \; \mathrm{arg} \; \min_{\bm{\Delta} \in \mathbb{R}^p} \; \left\{ \norm{\mathbf{r} \; + \; \XI \bm{\Delta} }_2^2 \; + \; \lambda_{\ell^2} \norm{\bm{\Delta}}_2^2\right\}, 
\end{equation*}
for which a closed-form solution of the same computational complexity as least squares exists. 
We refer to the weight matrix estimate obtained by using $\Deltaest^{\ell^2}$ in place of $\Deltaest$ in~\cref{eq:opt_matrix_estimator} as~$\Wweak$. By~\cref{prop:weak_weight_consistency}, we still obtain the same limiting guarantees of~\cref{theo:no_bias_in_the_limit} for $\Wweak$, as long as $\lambda_{\ell^2}$ is fixed ($\lambda_{\ell^2}$ is independent of $m$, $\XP$, $\YP$).
\begin{proposition}\label{prop:weak_weight_consistency}
Let $\lim_{m\to\infty} \frac{n(m)}{m} = c$ and $\lambda_{\ell^2} > 0$ be fixed. Then,
    \begin{equation*}
        \ulim{m} \; \MSE \left( \alphaWeight{\Wweak} \right) = 0.
    \end{equation*}
\end{proposition}
The proof for Proposition \ref{prop:weak_weight_consistency} is given in App.~\ref{sec:proof_weak_weight_consistency}.
\paragraph{Small $\norm{\bm{\Delta}}_0$.}
In other settings, we may have prior beliefs that only some treatment variables~$X_i$ are confounded, i.e., that the number of nonzero elements of $\bm{\Delta}$, denoted by $\norm{\bm{\Delta}}_0$, is small.
If we are unaware of
which treatments are confounded, but
$p$ is small, we can simply fit all $2^p$ possible models or use best subset selection~\citep[p. 205]{james2013introduction}. For larger $p$, a more efficient technique known as the LASSO employs $\ell^1$-regularization and has become a standard tool~\citep{tibshirani1996regression}. For the LASSO, approximate optimization techniques exist that have a computational complexity of $\mathcal{O}(p^2 n)$ \citep{efron2004lars}, which is of the same order as ordinary least squares. In this case, we reformulate
\eqref{equ:optimization} as
\begin{equation*}
    \Deltaest^{\ell^1} \; \gets \; \mathrm{arg} \; \min_{\bm{\Delta} \in \mathbb{R}^p} \; \left\{ \norm{ \mathbf{r} \; + \; \XI \bm{\Delta} }_2^2 \; + \; \lambda_{\ell^1} \norm{\bm{\Delta}}_1\right\},
\end{equation*}
for some $\lambda_{\ell^1} > 0$, and where $\norm{\; \cdot \;}_1$ denotes the $\ell^1$-norm. We refer to the weight matrix obtained by using $\Deltaest^{\ell^1}$ in place of $\Deltaest$ in~\cref{eq:opt_matrix_estimator} as~$\Wsparse$. 

\section{Experiments}\label{sec:experiments}
\begin{table*}[t]
    \centering
    \caption{\textbf{Mean squared error for the causal effect parameter $\alphab$ using various weighting schemes for different types of confounding.} 
    The standard plug-in optimal weight matrix estimator $\Wopthat$ generally does not perform well, while $\Wweak$ and $\Wsparse$, which benefit from prior knowledge, outperform prior work. Note that $\Wopt$ is an oracle that is generally not computable in practice. Numbers correspond to mean $\pm$ std.\ dev.\ over 1000 runs; the best method is highlighted in bold.}
    \label{tab:sim_experiment_outcomes}
    \resizebox{\textwidth}{!}{
    \small

    \caption{\textbf{Performance for varying dataset sizes  and ratios.}
    \textit{(Left)} All methods improve as the amount of 
    data is increased. More sophisticated weighting schemes outperform the purely interventional ($\WI$) and plug in estimators ($\Wopthat$) (not depicted due to close performance overlap with pure interventional), whereas data pooling ($\WP$) works well only for small $m$ and $\gamma$.
    \textit{(Right)} When keeping $m$ fixed and adding more observational data, $\Wweak$ clearly works best in strongly confounded ($\gamma=5$) settings. MSE and $\frac{n}{m}$ are plotted on a $\text{log}_{10}$ scale. Shaded areas indicate $\pm 0.5$ standard deviations. 
    }
    \label{fig:plots}
\end{figure*}
We investigate the empirical behavior of our proposed matrix weighted estimators in a finite sample setting and compare them with baselines and existing methods through simulations on synthetic data.\footnote{The source code for all experiments is available at: \href{https://github.com/rudolfwilliam/matrix_weighted_linear_estimators}{https://github.com/rudolfwilliam/matrix\_weighted\_linear\_estimators}}
To this end, we consider different experimental settings in which we vary the strength and sparsity of confounding, as well as the ratio and absolute quantity of observational and interventional data. 

\paragraph{Compared Methods.}
We report the mean squared error attained by the theoretically optimal weight matrix $\Wopt$ from~\eqref{equ:theor_optimal_matrix} as an oracle, as well as the plug-in estimator~$\Wopthat$ thereof from~\eqref{eq:opt_matrix_estimator}, and the regularized regression-based $\Wweak$ and $\Wsparse$ from~\cref{sec:inductive_biases}. 
For the latter two, we choose the regularization hyperparameters $\lambda_{\ell^2}$ and $\lambda_{\ell^1}$ by cross-validation on the interventional data. 
As baselines, we consider only using interventional data~($\WI=\Ib_p$) and data pooling according to $\WP$ from~\eqref{eq:W_data_pooling}.
We also compare to the~\citet{rosenman2020combining} scalar weighting scheme which was proposed for vectors of binary treatment effects and is given by $\Wb=\wrose\Ib_p$ with
\begin{equation*}
    \widehat{w}^m_{\text{rm}} \; := \; \max \left\{ 1 \; - \; \frac{\text{Tr} \left( \widehat{\textbf{Cov}}\left(\alphaI\right) \right)}{\norm{ \alphaI - \alphaO }_2^2}, \; 0 \right\}.
\end{equation*}
We emphasize that other commonly used methods for causal effect estimation from observational data such as propensity score matching~\citep{imai2004treatmentregimes} are not applicable, 
because they require the relevant confounders to be observed, which is not the case in our setting.

\paragraph{General Setup.} In all experiments, we use $p=30$ treatments, a one-dimensional ($d=1$) confounder~$Z$, and unit/isotropic (co)variances: $\sigma_{N_Y}^2= \sigma_{N_Z}^2= 1$, $\Sigmab_{\Nb_\Xb}=\Ib_p$. We sample $\tilde{\Nb}_{\Xb} \sim \mathcal{N}(\mathbf{0}, \COV(\XO))$,
$\bm{\alpha} \sim \mathcal{N}(\mathbf{0}, 9\mathbf{I}_p)$, and choose $\bb$ and $\gamma$ depending on the settings described below.
Unless otherwise specified, we then draw $m = 300$ interventional and $n = 600$ observational examples from $\Pint$ and $\Pobs$, respectively, and compute estimates of $\bm{\alpha}$ using the different weighting approaches. We repeat this procedure $1000$ times and report the resulting mean and standard deviation of the mean squared error.

\paragraph{Different Types of Confounding.}
In our main experiment, we investigate how estimators perform under different types of confounding encoded by ~\eqref{equ:SA_X} and~\eqref{eq:outcome}, specifically by the parameters $\bb\in\RR^p$ and $\gamma\in \RR$ (for a scalar confounder $Z$). 
For \textit{spread} confounding, we sample $\bb \sim \mathcal{N}(\mathbf{0}, \mathbf{I}_p)$ such that the confounder affects all treatment variables almost surely.
For \textit{sparse} confounding, we sample $b^{(k)}\sim \Ncal(0,1)$ for $k=1,...,5$, and $b^{(k)}=0$ otherwise,  such that only the first five treatments are confounded.
In both cases, we investigate $\gamma\in\{1,5\}$ which controls the strength of $Z\to~Y$ and thus the extent to which $\bm{\Delta} = 0$ is violated.

\paragraph{Main Results.} The results are presented in~\cref{tab:sim_experiment_outcomes}.
We find that our regularized estimators generally perform well, particularly when the underlying assumptions are satisfied: under sparse confounding $\Wsparse$ works best, and in the spread confounding case $\Wweak$ is only narrowly outperformed by $\wrose$ and $\WP$ when $\gamma=1$. Data pooling works relatively well when $\gamma=1$ (compared to $\gamma=5$) where the  violation of the identically distributed assumption is weak and the variance from estimating unknown quantities is not compensated by the bias reduction.
In contrast, both the purely interventional approach $\WI$ and the plug-in estimator $\Wopthat$ do not perform very well in this finite sample setting due to high variance, as explained in~\cref{sec:inductive_biases}.

\paragraph{Varying Data Set Sizes and Ratios.}
In~\cref{fig:plots}, we investigate how the different estimators behave across different data set sizes and ratios for the spread confounding setting.
In the left two plots, we vary the amount of interventional data $m$ while fixing the amount of observational data to $n=3m$. 
The results confirm our theoretical results: For small data set sizes, data pooling is a worthwhile alternative to more sophisticated weights, in particular if the violation against the assumption of identical distribution is minor ($\gamma = 1$). However, for large enough data set sizes, the approaches from both previous work and ours achieve a better score. Particularly, we see that $\Wweak$ outperforms all other weights in both scenarios for large enough data sets.

In the right two plots, we keep $m=500$ fixed and change~$n$ and thus the ratio of interventional to observational data. Unsurprisingly, we find that the mean squared error of $\WI$ remains constant. For strong confounding ($\gamma = 5$), we see that $\Wweak$ adapts best with a considerable margin: Unlike $\wrose$, it explicitly takes into account (an estimate of) the covariance structure of $\alphaO$  in constructing the weight matrix.
\section{Discussion}\label{sec:other_fields}
\paragraph{Connection to Transfer Learning.}\label{sec:transfer_learning}
Our setting bears resemblance to transfer and multi-task learning~\citep{thrun1995learning,caruana1997multitask}, specifically to supervised domain adaptation, which aims to leverage knowledge from a source domain to improve a model in a target domain, for which typically much less data is available.
In our case, we aim to use the source model  $\alphaO$, learned by estimating $\EE[Y|\Xb=\xb]$ in the observational setting, to improve our (high-variance) target model $\alphaI$ of $\EE[Y|\text{do}(\Xb\leftarrow\xb)]$.
Transfer learning 
can only work if the domains are sufficiently similar, resulting in numerous approaches leveraging different assumptions about shared components~\citep{quinonero2008dataset}. These assumptions are often phrased in causal terms~\citep{scholkopf2012causal,zhang2013domain,gong2016domain,rojas2018invariant}.
Similarly, our observational (source) and interventional (target) domains share the same causal model and only differ in the treatment assignment mechanisms~\eqref{equ:SA_X} and~\eqref{equ:replaced_SA}. Still, the bias in~\eqref{equ:Delta} can in theory be arbitrary large, and our methods from~\cref{sec:inductive_biases} implicitly rely on it being small or sparse.

\paragraph{Beyond Linear Regression.}
Some of our derivations and theoretical results rely on the fact that the confounding bias in~\eqref{eq:pert_model} is linear in $\xb$. 
For the class of \textit{linear} SCMs~\eqref{equ:multivarSCMconfounder}--\eqref{eq:outcome}, Gaussianity is necessary and sufficient\footnote{Note $\EE[Y|\Xb]=\gammab^\top\EE[\Zb|\Xb]+\alphab^\top\Xb$ and $\EE[\Zb|\Xb]$ is linear in $\Xb$ only in the Gaussian case~\citep[][Thm.~4.2]{peters2017elements}.} for this condition to hold, but it may also hold for more general classes of SCMs. 
For binary treatments $\Xb\in\{0,1\}^p$, in particular, it is always possible to write the difference between the biased and unbiased average treatment effect estimates using a constant offset $\Deltab$ akin to \eqref{eq:delta_plugin}, irrespective of the confounding relationship.\footnote{Specifically, we have $\Deltab= \EE [Y | \Xb = \mathbf{1} ] - \EE [Y | \Xb = \mathbf{0}] - \left( \EE [Y | \text{do}(\Xb \leftarrow \mathbf{1})] - \EE [Y | \text{do}(\Xb \leftarrow \mathbf{0}) ] \right)$.}
Future work may thus investigate nonlinear extensions, e.g., by drawing inspiration from semi-parametrics~\citep{robins1995semiparametric}, doubly robust estimation~\citep{bang2005doubly}, and debiased machine learning~\citep{chernozhukov2018double}.

\paragraph{Incorporating Covariates.}
Our current formulation does not \textit{explicitly} account for
observed confounders, or pre-treatment covariates, which need to be adjusted for in the observational setting to avoid introducing further bias. 
In principle, such covariates can simply be included in~$\Xb$, as different treatment components $X_i$ are allowed to be dependent.
However, this may result in high-dimensional treatments and thus render full randomization in~\eqref{equ:replaced_SA} unrealistic.
Other covariates, while unproblematic with regard to bias, may help further reduce variance~\citep{henckel2022graphical}. 
Extending our framework to incorporate different types of covariates is thus a worthwhile future direction.

\section{Conclusion}\label{sec:conclusion}
In the present work, we have introduced a new class of matrix weighted linear estimators for learning causal effects of continuous treatments from finite observational and interventional data. 
Here, our focus has been on optimizing statistical efficiency, which complements the vast causal inference literature on identification from heterogeneous data. 
Our estimators are connected to classical ideas from shrinkage estimation applied to causal learning and provide a unifying account of data pooling and ridge regression, which emerge as special cases.
We show that our estimators are theoretically grounded and compare favorably to baselines and prior work in simulations.
While we restricted our analysis to linear models for now, we hope that the insights and methods developed here will also be useful for a broader class of causal models and transfer learning problems.

\begin{acknowledgements}
    We thank the anonymous reviewers for useful comments and suggestions that helped improve the manuscript.

    We thank the Branco Weiss Fellowship, administered by ETH Zurich, for the support.
    This work was further supported by the T\"ubingen AI Center and by the German Research Foundation (DFG) under Germany’s excellence strategy – EXC number 2064/1 – project number 390727645.
\end{acknowledgements}

{%
\bibliography{kladny_306}

\begin{thebibliography}{58}
\providecommand{\natexlab}[1]{#1}
\providecommand{\url}[1]{\texttt{#1}}
\expandafter\ifx\csname urlstyle\endcsname\relax
  \providecommand{\doi}[1]{doi: #1}\else
  \providecommand{\doi}{doi: \begingroup \urlstyle{rm}\Url}\fi

\bibitem[Angrist and Pischke(2009)]{angrist2009mostly}
J.~D. Angrist and J.-S. Pischke.
\newblock \emph{{Mostly harmless econometrics: An empiricist's companion}}.
\newblock Princeton University Press, 2009.

\bibitem[Angrist et~al.(1996)Angrist, Imbens, and Rubin]{angrist1996identification}
J.~D. Angrist, G.~W. Imbens, and D.~B. Rubin.
\newblock {Identification of causal effects using instrumental variables}.
\newblock \emph{Journal of the American Statistical Association}, 91\penalty0 (434):\penalty0 444--455, 1996.

\bibitem[Bang and Robins(2005)]{bang2005doubly}
H.~Bang and J.~M. Robins.
\newblock Doubly robust estimation in missing data and causal inference models.
\newblock \emph{Biometrics}, 61\penalty0 (4):\penalty0 962--973, 2005.

\bibitem[Bareinboim and Pearl(2012)]{bareinboim2012causal}
E.~Bareinboim and J.~Pearl.
\newblock {Causal Inference by Surrogate Experiments: z-Identifiability}.
\newblock In \emph{{Proceedings of the 28th Conference on Uncertainty in Artificial Intelligence}}, pages 113--120, 2012.

\bibitem[Bareinboim and Pearl(2016)]{bareinboim2016causal}
E.~Bareinboim and J.~Pearl.
\newblock {Causal inference and the data-fusion problem}.
\newblock \emph{Proceedings of the National Academy of Sciences}, 113\penalty0 (27):\penalty0 7345--7352, 2016.

\bibitem[Caruana(1997)]{caruana1997multitask}
R.~Caruana.
\newblock {Multitask Learning}.
\newblock \emph{Machine Learning}, 28\penalty0 (1):\penalty0 41--75, 1997.

\bibitem[{\'C}evid et~al.(2020){\'C}evid, B{\"u}hlmann, and Meinshausen]{cevid2020spectral}
D.~{\'C}evid, P.~B{\"u}hlmann, and N.~Meinshausen.
\newblock {Spectral Deconfounding via Perturbed Sparse Linear Models}.
\newblock \emph{The Journal of Machine Learning Research}, 21\penalty0 (1):\penalty0 9442--9482, 2020.

\bibitem[Cheng and Cai(2021)]{cheng2021adaptive}
D.~Cheng and T.~Cai.
\newblock {Adaptive Combination of Randomized and Observational Data}.
\newblock \emph{arXiv:2111.15012}, 2021.

\bibitem[Chernozhukov et~al.(2018)Chernozhukov, Chetverikov, Demirer, Duflo, Hansen, Newey, and Robins]{chernozhukov2018double}
V.~Chernozhukov, D.~Chetverikov, M.~Demirer, E.~Duflo, C.~Hansen, W.~Newey, and J.~Robins.
\newblock Double/debiased machine learning for treatment and structural parameters: Double/debiased machine learning.
\newblock \emph{The Econometrics Journal}, 21\penalty0 (1), 2018.

\bibitem[Colnet et~al.(2020)Colnet, Mayer, Chen, Dieng, Li, Varoquaux, Vert, Josse, and Yang]{colnet2020causal}
B.~Colnet, I.~Mayer, G.~Chen, A.~Dieng, R.~Li, G.~Varoquaux, J.-P. Vert, J.~Josse, and S.~Yang.
\newblock {Causal inference methods for combining randomized trials and observational studies: a review}.
\newblock \emph{arXiv:2011.08047}, 2020.

\bibitem[Correa and Bareinboim(2020)]{correa2020general}
J.~Correa and E.~Bareinboim.
\newblock {General transportability of soft interventions: Completeness results}.
\newblock \emph{Advances in Neural Information Processing Systems}, 33:\penalty0 10902--10912, 2020.

\bibitem[Eberhardt and Scheines(2007)]{eberhardt2007interventions}
F.~Eberhardt and R.~Scheines.
\newblock {Interventions and causal inference}.
\newblock \emph{Philosophy of science}, 74\penalty0 (5):\penalty0 981--995, 2007.

\bibitem[Efron(2012)]{efron2012large}
B.~Efron.
\newblock \emph{{Large-scale inference: empirical {B}ayes methods for estimation, testing, and prediction}}, volume~1.
\newblock Cambridge University Press, 2012.

\bibitem[Efron and Morris(1973)]{efron1973stein}
B.~Efron and C.~Morris.
\newblock {Stein's estimation rule and its competitors—an empirical Bayes approach}.
\newblock \emph{Journal of the American Statistical Association}, 68\penalty0 (341):\penalty0 117--130, 1973.

\bibitem[Efron et~al.(2004)Efron, Hastie, Johnstone, and Tibshirani]{efron2004lars}
B.~Efron, T.~Hastie, I.~Johnstone, and R.~Tibshirani.
\newblock {Least Angle Regression}.
\newblock \emph{The Annals of Statistics}, 32\penalty0 (2), 2004.

\bibitem[Fisher(1936)]{fisher1936design}
R.~A. Fisher.
\newblock {Design of experiments}.
\newblock \emph{British Medical Journal}, 1\penalty0 (3923):\penalty0 554, 1936.

\bibitem[Gong et~al.(2016)Gong, Zhang, Liu, Tao, Glymour, and Sch{\"o}lkopf]{gong2016domain}
M.~Gong, K.~Zhang, T.~Liu, D.~Tao, C.~Glymour, and B.~Sch{\"o}lkopf.
\newblock Domain adaptation with conditional transferable components.
\newblock In \emph{International Conference on Machine Learning}, pages 2839--2848, 2016.

\bibitem[Green and Strawderman(1991)]{strawderman1991stein}
E.~J. Green and W.~E. Strawderman.
\newblock {A James-Stein Type Estimator for Combining Unbiased and Possibly Biased Estimators}.
\newblock \emph{Journal of the American Statistical Association}, 86\penalty0 (416):\penalty0 1001--1006, 1991.

\bibitem[Green et~al.(2005)Green, Strawderman, Amateis, and Reams]{green2005improved}
E.~J. Green, W.~E. Strawderman, R.~L. Amateis, and G.~A. Reams.
\newblock {Improved Estimation for Multiple Means with Heterogeneous Variances}.
\newblock \emph{Forest Science}, 51\penalty0 (1):\penalty0 1--6, 2005.

\bibitem[Hastie et~al.(2009)Hastie, Tibshirani, and Friedman]{hastie2009elements}
T.~Hastie, R.~Tibshirani, and J.~Friedman.
\newblock \emph{{The Elements of Statistical Learning}}.
\newblock Springer, 2009.

\bibitem[Hatt et~al.(2022)Hatt, Berrevoets, Curth, Feuerriegel, and van~der Schaar]{hatt2022combining}
T.~Hatt, J.~Berrevoets, A.~Curth, S.~Feuerriegel, and M.~van~der Schaar.
\newblock {Combining observational and randomized data for estimating heterogeneous treatment effects}.
\newblock \emph{arXiv:2202.12891}, 2022.

\bibitem[Henckel et~al.(2022)Henckel, Perkovi{\'c}, and Maathuis]{henckel2022graphical}
L.~Henckel, E.~Perkovi{\'c}, and M.~H. Maathuis.
\newblock Graphical criteria for efficient total effect estimation via adjustment in causal linear models.
\newblock \emph{Journal of the Royal Statistical Society Series B}, 84\penalty0 (2):\penalty0 579--599, 2022.

\bibitem[Hern\'an and Robins(2020)]{hernan2020causal}
M.~A. Hern\'an and J.~M. Robins.
\newblock \emph{{Causal inference: What if}}.
\newblock Boca Raton: Chapman \& Hall/CRC, 2020.

\bibitem[Hoerl(1970)]{hoerl1970ridge}
A.~E. Hoerl.
\newblock {Ridge Regression: Biased Estimation for Nonorthogonal Problems}.
\newblock \emph{Technometrics}, 12\penalty0 (1):\penalty0 55--67, 1970.

\bibitem[Huang and Valtorta(2006)]{huang2006identifiability}
Y.~Huang and M.~Valtorta.
\newblock {Identifiability in Causal Bayesian Networks: A Sound and Complete Algorithm}.
\newblock In \emph{{Proceedings of the National Conference on Artificial Intelligence}}, volume~21, pages 1149--1154, 2006.

\bibitem[Ilse et~al.(2021)Ilse, Forré, Welling, and Mooij]{ilse2021combining}
M.~Ilse, P.~Forré, M.~Welling, and J.~M. Mooij.
\newblock {Combining Interventional and Observational Data Using Causal Reductions}.
\newblock \emph{arXiv:2103.04786}, pages 1--42, 2021.

\bibitem[Imai and Dyk(2004)]{imai2004treatmentregimes}
K.~Imai and D.~A.~V. Dyk.
\newblock {Causal Inference with General Treatment Regimes: Generalizing the Propensity Score}.
\newblock \emph{Journal of the American Statistical Association}, 99\penalty0 (467):\penalty0 854--866, 2004.

\bibitem[Imbens and Rubin(2015)]{imbens2015causal}
G.~W. Imbens and D.~B. Rubin.
\newblock \emph{{Causal inference in statistics, social, and biomedical sciences}}.
\newblock Cambridge University Press, 2015.

\bibitem[James et~al.(2013)James, Witten, Hastie, and Tibshirani]{james2013introduction}
G.~James, D.~Witten, T.~Hastie, and R.~Tibshirani.
\newblock \emph{{An Introduction to Statistical Learning}}.
\newblock Springer, 2013.

\bibitem[James and Stein(1961)]{james1961estimation}
W.~James and C.~Stein.
\newblock {Estimation with Quadratic Loss}.
\newblock In \emph{{Proceedings of the 4th Berkeley Symposium on Probability and Statistics}}. Berkeley, CA: University of California Press, 1961.

\bibitem[Kallus et~al.(2018)Kallus, Puli, and Shalit]{kallus2018removing}
N.~Kallus, A.~M. Puli, and U.~Shalit.
\newblock {Removing hidden confounding by experimental grounding}.
\newblock \emph{Advances in Neural Information Processing Systems}, 31, 2018.

\bibitem[Lee et~al.(2020)Lee, Correa, and Bareinboim]{sanghack2020gID}
S.~Lee, J.~D. Correa, and E.~Bareinboim.
\newblock {General Identifiability with Arbitrary Surrogate Experiments}.
\newblock In \emph{Proceedings of the 35th Uncertainty in Artificial Intelligence Conference}, pages 389--398, 2020.

\bibitem[Morgan and Winship(2014)]{morgan2014counterfactuals}
S.~L. Morgan and C.~Winship.
\newblock \emph{{Counterfactuals and Causal Inference: Methods and Principles for Social Research}}.
\newblock Cambridge University Press, 2014.

\bibitem[Neyman(1923)]{neyman1923application}
J.~Neyman.
\newblock {On the application of probability theory to agricultural experiments: essay on principles}.
\newblock \emph{Statistical Science}, 5:\penalty0 465--480, 1923.

\bibitem[Pearl(1995)]{pearl1995causal}
J.~Pearl.
\newblock {Causal diagrams for empirical research}.
\newblock \emph{Biometrika}, 82\penalty0 (4):\penalty0 669--688, 1995.

\bibitem[Pearl(2009)]{pearl2009causality}
J.~Pearl.
\newblock \emph{{Causality: models, reasoning, and inference}}.
\newblock Cambridge University Press, 2nd edition, 2009.

\bibitem[Pearl and Bareinboim(2014)]{pearl2014external}
J.~Pearl and E.~Bareinboim.
\newblock {External Validity: From Do-Calculus to Transportability Across Populations}.
\newblock \emph{Statistical Science}, 29\penalty0 (4):\penalty0 579--595, 2014.

\bibitem[Peters et~al.(2017)Peters, Janzing, and Sch{\"o}lkopf]{peters2017elements}
J.~Peters, D.~Janzing, and B.~Sch{\"o}lkopf.
\newblock \emph{{Elements of Causal Inference: Foundations and Learning Algorithms}}.
\newblock MIT Press, 2017.

\bibitem[Qui{\~n}onero-Candela et~al.(2008)Qui{\~n}onero-Candela, Sugiyama, Schwaighofer, and Lawrence]{quinonero2008dataset}
J.~Qui{\~n}onero-Candela, M.~Sugiyama, A.~Schwaighofer, and N.~D. Lawrence.
\newblock \emph{{Dataset Shift in Machine Learning}}.
\newblock MIT Press, 2008.

\bibitem[Reichenbach(1956)]{reichenbach1956direction}
H.~Reichenbach.
\newblock \emph{{The Direction of Time}}, volume~65.
\newblock University of California Press, 1956.

\bibitem[Robbins(1964)]{robbins1964empirical}
H.~Robbins.
\newblock {The empirical {B}ayes approach to statistical decision problems}.
\newblock \emph{The Annals of Mathematical Statistics}, 35\penalty0 (1):\penalty0 1--20, 1964.

\bibitem[Robins and Rotnitzky(1995)]{robins1995semiparametric}
J.~M. Robins and A.~Rotnitzky.
\newblock Semiparametric efficiency in multivariate regression models with missing data.
\newblock \emph{Journal of the American Statistical Association}, 90\penalty0 (429):\penalty0 122--129, 1995.

\bibitem[Rojas-Carulla et~al.(2018)Rojas-Carulla, Sch{\"o}lkopf, Turner, and Peters]{rojas2018invariant}
M.~Rojas-Carulla, B.~Sch{\"o}lkopf, R.~Turner, and J.~Peters.
\newblock {Invariant Models for Causal Transfer Learning}.
\newblock \emph{The Journal of Machine Learning Research}, 19\penalty0 (1):\penalty0 1309--1342, 2018.

\bibitem[Rosenman et~al.(2020)Rosenman, Basse, Owen, and Baiocchi]{rosenman2020combining}
E.~Rosenman, G.~Basse, A.~Owen, and M.~Baiocchi.
\newblock {Combining observational and experimental datasets using shrinkage estimators}.
\newblock \emph{Biometrics}, 2020.

\bibitem[Rosenman et~al.(2022)Rosenman, Owen, Baiocchi, and Banack]{rosenman2022propensity}
E.~T. Rosenman, A.~B. Owen, M.~Baiocchi, and H.~R. Banack.
\newblock {Propensity score methods for merging observational and experimental datasets}.
\newblock \emph{Statistics in Medicine}, 41\penalty0 (1):\penalty0 65--86, 2022.

\bibitem[Rubin(1974)]{rubin1974estimating}
D.~B. Rubin.
\newblock {Estimating causal effects of treatments in randomized and nonrandomized studies.}
\newblock \emph{Journal of educational Psychology}, 66\penalty0 (5):\penalty0 688--701, 1974.

\bibitem[Schaffer(1993)]{schaffer1993selecting}
C.~Schaffer.
\newblock {Selecting a Classification Method by Cross-Validation}.
\newblock \emph{Machine Learning}, 13:\penalty0 135--143, 1993.

\bibitem[Sch{\"o}lkopf et~al.(2012)Sch{\"o}lkopf, Janzing, Peters, Sgouritsa, Zhang, and Mooij]{scholkopf2012causal}
B.~Sch{\"o}lkopf, D.~Janzing, J.~Peters, E.~Sgouritsa, K.~Zhang, and J.~M. Mooij.
\newblock {On Causal and Anticausal Learning}.
\newblock In \emph{{International Conference on Machine Learning}}, 2012.

\bibitem[Shpitser and Pearl(2006)]{shpitser2006identification}
I.~Shpitser and J.~Pearl.
\newblock {Identification of Joint Interventional Distributions in Recursive Semi-Markovian Causal Models}.
\newblock In \emph{{Proceedings of the National Conference on Artificial Intelligence}}, volume~21, pages 1219--1226, 2006.

\bibitem[Spirtes et~al.(2000)Spirtes, Glymour, and Scheines]{spirtes2000pc}
P.~Spirtes, C.~Glymour, and R.~Scheines.
\newblock \emph{{Causation, Prediction, and Search}}.
\newblock MIT Press, 2000.

\bibitem[Stein(1956)]{stein1956inadmissibility}
C.~Stein.
\newblock {Inadmissibility of the Usual Estimator for the Mean of a Multivariate Normal Distribution}.
\newblock In \emph{{Proceedings of the third Berkeley Symposium on Mathematical Statistics and Probability}}, volume~3, pages 197--207. University of California Press, 1956.

\bibitem[Thrun(1995)]{thrun1995learning}
S.~Thrun.
\newblock {Is Learning The n-th Thing Any Easier Than Learning The First?}
\newblock \emph{Advances in Neural Information Processing Systems}, 8, 1995.

\bibitem[Tian and Pearl(2002)]{tian2002general}
J.~Tian and J.~Pearl.
\newblock {A general identification condition for causal effects}.
\newblock In \emph{{Proceedings of the AAAI Conference on Artificial Intelligence}}, pages 567--573, 2002.

\bibitem[Tibshirani(1996)]{tibshirani1996regression}
R.~Tibshirani.
\newblock {Regression Shrinkage and Selection via the Lasso}.
\newblock \emph{Journal of the Royal Statistical Society: Series B (Methodological)}, 58\penalty0 (1):\penalty0 267--288, 1996.

\bibitem[Wang et~al.(2017)Wang, Solus, Yang, and Uhler]{wang2017permutation}
Y.~Wang, L.~Solus, K.~Yang, and C.~Uhler.
\newblock {Permutation-based Causal Inference Algorithms with Interventions}.
\newblock \emph{Advances in Neural Information Processing Systems}, 30, 2017.

\bibitem[Wasserman(2006)]{wasserman2006allofnonparam}
L.~Wasserman.
\newblock \emph{{All of Nonparametric Statistics}}.
\newblock Springer, 2006.

\bibitem[Yang and Ding(2020)]{pmid33088006}
S.~Yang and P.~Ding.
\newblock {{C}ombining {M}ultiple {O}bservational {D}ata {S}ources to {E}stimate {C}ausal {E}ffects}.
\newblock \emph{Journal of the American Statistical Association}, 115\penalty0 (531):\penalty0 1540--1554, 2020.

\bibitem[Zhang et~al.(2013)Zhang, Sch{\"o}lkopf, Muandet, and Wang]{zhang2013domain}
K.~Zhang, B.~Sch{\"o}lkopf, K.~Muandet, and Z.~Wang.
\newblock {Domain Adaptation under Target and Conditional Shift}.
\newblock In \emph{International Conference on Machine Learning}, pages 819--827, 2013.

\end{thebibliography}
}

\pagebreak
\onecolumn

\begin{center}
    \huge\textbf{Appendix}\\
\end{center}

\setcounter{equation}{0}
\def\theequation{\AlphAlph{\value{equation}}}
\setcounter{section}{0}
\newcommand{\liminfm}{\underset{m \rightarrow \infty}{\text{lim} \, \text{inf}}} %

\newcommand{\swap}[3][-]{#3#1#2} %

\newtheorem*{prospec}{Proposition 4.2}
\newtheorem{lemma2}{Lemma}[section]

\renewcommand\thesection{\Alph{section}}
\section{Proofs}

\subsection{Proposition \ref{prop:MSE_greater_zero}} \label{sec:proof_prop_MSE_greater_zero}

\begin{proof}
    We begin by observing that we can write $\WP$ as

    \begin{equation} \label{equ:rewrite_pooling_matrix}
        \WP = \left(m^{-1} \XI^{\top} \XI \; + \; \frac{n}{m} n^{-1} \XO^{\top} \XO \right)^{-1} \left( m^{-1} \XI^{\top} \XI \right).
    \end{equation}
    We apply the strong law of large numbers to obtain that 
    \begin{equation*}
        m^{-1} \XI^{\top} \XI \asconv \COV(\XI) \quad \text{and} \quad n^{-1} \XO^{\top} \XO \asconv \COV(\XO).
    \end{equation*}
    Due to the fact that $\ulim{m} \frac{n(m)}{m}=c$ for some $c>0$, we conclude

    \begin{equation*}
        \WP \; \asconv \; \Wb_{\infty} \; := \; \left( \text{\textbf{Cov}}(\XI) \; + \; c \cdot \text{\textbf{Cov}}(\XO)\right)^{-1} \text{\textbf{Cov}}(\XI).
    \end{equation*}

    We observe that 

    \begin{equation*}
        \left( \Ib - \Wb_{\infty}\right) = \left( \COV(\XI) + c \cdot \COV(\XO) \right)^{-1} c \cdot \COV(\XO).
    \end{equation*}
    Since both covariance matrices are positive definite, so is $\COV(\XI) + c \cdot \COV(\XO)$. We conclude that the smallest singular value of $\Ib - \Wb_{\infty}$ is strictly greater than 0. This means
    \begin{equation*}
        \big| \big| \mathbb{E}[\alphaWeight{\Wb_{\infty}}] - \bm{\alpha} \big| \big|_2^2 \; = \; || \left( \mathbf{I}_p - \Wb_{\infty} \right) \bm{\Delta} ||_2^2 \; \geq \; c' || \bm{\Delta} ||_2^2,
    \end{equation*}

    for some fixed constant $c' > 0$. We obtain therefore
    
    \begin{equation*}
        0 < \underset{m \rightarrow \infty}{\text{lim}} \big| \big| \mathbb{E}[\alphaWeight{\Wb_{\infty}}] - \bm{\alpha} \big| \big|_2^2 \leq \underset{m \rightarrow \infty}{\text{lim}} \; \text{MSE} \, \left( \alphaWeight{\Wb_{\infty}} \right),
    \end{equation*}
    where we invoked Jensen's inequality. We see that $\Wb_{\infty}$ is constant and bounded. We note that almost sure convergence implies convergence in probability. We can thus apply Lemma \ref{lem:MSEconvergence}, which yields the desired result
    \begin{equation*}
        0 < \underset{m \rightarrow \infty}{\text{lim}} \; \text{MSE} \, \left( \alphaWeight{\Wb_{\infty}} \right) \leq \underset{m \rightarrow \infty}{\text{lim}} \; \text{MSE} \, \left( \alphaWeight{\WP} \right).
    \end{equation*}
\end{proof}

\subsection{Proposition 4.2} \label{sec:prop_pool_vanishing_MSE}

\begin{prospec}
Let $\lim_{m\to\infty} \frac{n(m)}{m} = 0$. Then, it holds that
\begin{equation*}
    \lim_{m\to\infty} \MSE \left(\alphaP \right) = 0.
\end{equation*}
\end{prospec}

\begin{proof}
    Similar to the proof of Proposition \ref{prop:MSE_greater_zero}, we employ the formulation of \eqref{equ:rewrite_pooling_matrix} and consider the term
    \begin{equation*}
         \frac{n}{m} \cdot n^{-1} \XO^{\top} \XO.
    \end{equation*}
    We see that $\ulim{m}\frac{n(m)}{m} = 0$ and by the strong law of large numbers, $n^{-1} \XO^{\top} \XO \asconv \COV(\XO)$. Hence, we obtain that
    \begin{equation*}
        \frac{n}{m} \cdot n^{-1} \XO^{\top} \XO \asconv \mathbf{0}.
    \end{equation*}
    By the continuous mapping theorem, we conclude that
    \begin{equation*}
         \WP \asconv \mathbf{I}_p,
    \end{equation*}
    and by Lemma \ref{lem:upper_bounded}, this implies that
    \begin{equation*}
         \ulim{m} \; \text{MSE} \left( \alphaWeight{\WP} \right) \; \leq \; \ulim{m} \text{MSE} \left( \alphaI \right) \; = \; 0.
    \end{equation*}
\end{proof}

\subsection{Proposition \ref{prop:weight_matrix_cons}} \label{sec:proof_weight_matrix_cons}

\begin{proof}
    We rewrite $\Wopthat$ as follows:

    \begin{align*}
        \Wopthat = &\left( n^{-1} \left( n^{-1} \XO^{\top}\XO \right)^{-1} \hat{\sigma}^2_{Y|X} + \hat{\bm{\Delta}} \hat{\bm{\Delta}}^{\top} + \epsilon \Ib_p \right) \\ 
        &\left( n^{-1} \left( n^{-1} \XO^{\top}\XO \right)^{-1} \hat{\sigma}^2_{Y|X} + m^{-1} \left( m^{-1} \XI^{\top}\XI \right)^{-1} \hat{\sigma}^2_{Y|\text{do}(X)} + \hat{\bm{\Delta}} \hat{\bm{\Delta}}^{\top} + \epsilon \Ib_p \right)^{-1},
    \end{align*}
    where we insert any almost surely converging estimators for $\bm{\Delta}$, $\sigma_{Y|X}^2$ and $\sigma_{Y|\text{do}(X)}^2$ instead of their ground-truth values. By almost sure convergence of linear estimators individually, we see that this holds specifically for $\hat{\bm{\Delta}} = \alphaO - \alphaI$. Also, we can use the strong law of large numbers to conclude almost sure convergence of $\hat{\sigma}^2_{Y|X}$ and $\hat{\sigma}^2_{Y|\text{do}(X)}$.
    
    We now show $\Wopthat \asconv \Ib_p$: First, we see that
    \begin{equation*}
        (cm)^{-1} \left( n^{-1} \XO^{\top}\XO \right)^{-1} \hat{\sigma}^2_{Y|X} \asconv \mathbf{0} \quad \text{and} \quad m^{-1} \left( m^{-1} \XI^{\top}\XI \right)^{-1} \hat{\sigma}^2_{Y|\text{do}(X)} \asconv \mathbf{0},
    \end{equation*}
    since $m^{-1} \XI^{\top}\XI \; \hat{\sigma}^2_{Y|\text{do}(X)}$ and $n^{-1} \XO^{\top}\XO \; \hat{\sigma}^2_{Y|X}$ converge almost surely to constants and $m^{-1}$ vanishes. Hence,
    \begin{equation*}
        \Wopthat \asconv \left( \bm{\Delta} \bm{\Delta}^{\top} + \epsilon \Ib_p \right) \; \left( \bm{\Delta} \bm{\Delta}^{\top} + \epsilon \Ib_p \right)^{-1} = \Ib_p.
    \end{equation*}
\end{proof}

\subsection{Theorem \ref{theo:no_bias_in_the_limit}} \label{sec:proof_no_bias_in_the_limit}

\begin{proof}
We have that $\Ib_p$ is bounded in norm, almost surely. So we can apply Lemma \ref{lem:upper_bounded} to see that
    \begin{equation*}
        \ulim{m} \; \text{MSE} \, \big(\alphaWeight{\Wopthat}\big) \; \leq \; \ulim{m} \MSE \big( \alphaI \big) \; = \; 0.
    \end{equation*}
\end{proof}

\subsection{Proposition \ref{prop:weak_weight_consistency}} \label{sec:proof_weak_weight_consistency}
\begin{proof}
By Theorem~\ref{theo:no_bias_in_the_limit}, it suffices to show that $\Wweak \asconv \Ib_p$. Since the other quantities $\text{\textbf{Cov}}(\alphaI)$, $\text{\textbf{Cov}}(\alphaO)$ for estimating $\Wopt$ remain unchanged compared to $\Wopthat$, it suffices to show that the modified computation of $\Deltaest$ we call $\hat{\bm{\Delta}}_m^{\ell^2}$ converges almost surely to the true $\bm{\Delta} = \alphaItrue - \alphaOtrue$, where $\alphaItrue$ and $\alphaOtrue$ are short-hand for $\mathbb{E}_{\text{int}}[Y | \Xb = \xb]$ and $\mathbb{E}_{\text{obs}}[Y | \Xb = \xb]$, respectively. We observe that  $\hat{\bm{\Delta}}_m^{\ell^2}$ has a closed-form solution
\begin{align} 
    \hat{\bm{\Delta}}_m^{\ell^2} &= -(\XI^{\top}\XI + \lambda_{\ell^2} \mathbf{I}_p)^{-1} \XI^{\top} (\YI - \XI\alphaO) \\
    &= (\XI^{\top}\XI + \lambda_{\ell^2} \mathbf{I}_p)^{-1} \XI^{\top} \XI \alphaO - (\XI^{\top}\XI + \lambda_{\ell^2} \mathbf{I}_p)^{-1} \XI^{\top}\YI,\label{equ:weak_bias_cons}
\end{align}

since  $\alphaO$ is again a closed-form solution to an ordinary least squares problem. Considering the first term in \eqref{equ:weak_bias_cons}, we conclude almost sure convergence with respect to $\alphaItrue$ (it is simply the ridge regression solution on the interventional data, which is well-known to converge almost surely for fixed $\lambda_{\ell^2}$). The second term satisfies
\begin{equation*}
    (\XI^{\top}\XI + \lambda_{\ell^2} \mathbf{I}_p)^{-1} \XI^{\top} \XI \asconv \mathbf{I}_p \quad \text{and} \quad \alphaO \asconv \alphaOtrue.
\end{equation*}
This leads to the desired conclusion.
\end{proof}

\section{Additional Lemmas}

\begin{lemma2} \label{lem:MSEconvergence}
    Let $\West - \Wgeneric \cons \mathbf{0}$ \footnote{We note that $\Wgeneric$ may be random.} and let there exist $c > 0$, $m' \in \mathbb{N}$, such that $||\Wgeneric||_2 \leq c, \; \text{for all} \; m \geq m'$, almost surely. Then, it holds that
    \begin{equation*}
    \ulim{m} \text{MSE} \, \big( \alphaWeight{\Wgeneric} \big) \; \leq \; \ulim{m} \; \text{MSE} \, \big(\alphaWeight{\West}\big),
    \end{equation*}
    where $\cons$ denotes convergence in probability.
\end{lemma2}

\begin{proof}
    We derive a lower bound on $\text{MSE} \, \big(\alphaWeight{\West}\big)$ by using the formulation
    \begingroup
    \addtolength{\jot}{.5em}
    \begin{equation} \label{eq:upper_lower_bound}
    \begin{aligned}
        \text{MSE} \, \big(\alphaWeight{\West}\big) \; = \; &\mathbb{E}\left[ \mathbbm{1}\left\{||\West - \Wgeneric||_2 \leq \epsilon \right\} \; ||\alphaWeight{\West} - \bm{\alpha}||_2^2 \right] \; + \\ 
        &\mathbb{E} \left[ \mathbbm{1}\left\{||\West - \Wgeneric||_2 > \epsilon \right\} \;  ||\alphaWeight{\West} - \bm{\alpha}||_2^2 \right], \quad \forall \epsilon > 0.
    \end{aligned}
    \end{equation}
    \endgroup
    We bound the second summand of \eqref{eq:upper_lower_bound} from below by zero. For the first summand, we use reverse triangle inequality, which yields
   \begingroup
    \addtolength{\jot}{.5em}
    \begin{equation}\label{eq:first_summand_upper_bound}
        \begin{aligned}
        & &&\mathbb{E}\left[ \mathbbm{1}\left\{||\West - \Wgeneric||_2 \leq \epsilon \right\} ||\alphaWeight{\West} - \bm{\alpha}||_2^2 \right] \; = \; \mathbb{E}\left[ \mathbbm{1}\left\{||\West - \Wgeneric||_2 \leq \epsilon \right\} ||\alphaWeight{\West} - \alphaWeight{\Wgeneric} - (\bm{\alpha} - \alphaWeight{\Wgeneric} )||_2^2 \right] \\
        &\geq \; &&\mathbb{E}\left[\mathbbm{1}\left\{||\West - \Wgeneric||_2 \leq \epsilon \right\}||\alphaWeight{\Wgeneric} - \bm{\alpha}||_2^2\right] -  2\sqrt{\mathbb{E}\left[\mathbbm{1}\left\{||\West - \mathbf{W}||_2 \leq \epsilon \right\} ||\alphaWeight{\West} - \alphaWeight{\Wgeneric}||_2^2 \right] \; \mathbb{E}\left[||\alphaWeight{\Wgeneric} - \bm\alpha||_2^2\right]}\; + \\ 
        & &&\mathbb{E}\left[\mathbbm{1}\left\{|| \West - \Wgeneric ||_2 \leq \epsilon \right\} ||\alphaWeight{\West} - \alphaWeight{\Wgeneric}||_2^2\right] \\
        &\geq \; &&\text{MSE} (\alphaWeight{\Wgeneric}) - \mathbb{E}\left[\mathbbm{1}\left\{||\West - \Wgeneric||_2 > \epsilon \right\} ||\alphaWeight{\Wgeneric} - \bm{\alpha}||_2^2\right] - \\ 
        & &&2\sqrt{\mathbb{E}\left[\mathbbm{1}\left\{||\West - \Wgeneric||_2 \leq \epsilon \right\} ||\alphaWeight{\West} - \alphaWeight{\Wgeneric}||_2^2 \right] \; \mathbb{E}\left[||\alphaWeight{\Wgeneric} - \bm{\alpha}||_2^2\right]}.
        \end{aligned}
    \end{equation}
    \endgroup
    For any constant $\mathbf{W}, \Wb' \in \mathbb{R}^{p \times p}$, we rewrite
    \begingroup
    \addtolength{\jot}{.5em}
    \begin{equation*}
        \begin{aligned}
            \mathbb{E}\left[ ||\alphaWeight{\Wb'} - \alphaWeight{\Wb}||_2^2 \right] \; &= \; &&\mathbb{E}\left[||(\Wb' - \mathbf{W})\alphaI \; + \; (\mathbf{W} - \Wb')\alphaO||_2^2 \right] \\
            &\leq &&2 \left( ||\mathbf{W} - \Wb'||_2^2 \text{Tr}\left( \mathbb{E}\left[\alphaI \widehat{\bm{\alpha}}_{\textsc{i}}^{m \, \top}\right] \right)  \; + \; ||\mathbf{W} - \Wb'||_2^2 \text{Tr}\left( \mathbb{E}\left[\alphaO \widehat{\bm{\alpha}}_{\textsc{o}}^{n \, \top}\right] \right) \right) \\
            &= &&2 ||\mathbf{W} - \Wb'||_2^2 \Bigg[ \left( ||\mathbb{E}\left[\alphaI\right]||_2^2 \; + \; \text{Tr} \left( \text{\textbf{Cov}} \left(\alphaI \right) \right) \right) \; + \; \left( ||\mathbb{E}\left[\alphaO\right]||_2^2 \; + \; \text{Tr} \left( \text{\textbf{Cov}} \left(\alphaO \right) \right) \right) \Bigg],
        \end{aligned}
    \end{equation*}
    \endgroup
    where we have used Young's inequality in the first step. We see that both $||\mathbb{E}\left[\alphaI\right]||_2^2$ and $||\mathbb{E}\left[\alphaO\right]||_2^2$ remain bounded $\forall m$, while $\text{Tr} \left( \text{\textbf{Cov}} \left(\alphaO \right) \right)$ and $\text{Tr} \left( \text{\textbf{Cov}} \left(\alphaI \right) \right)$ decrease monotonically in $m$. Hence, we conclude that for any $\epsilon' > 0$, there exists an $\epsilon > 0$ such that
    \begingroup
    \addtolength{\jot}{.5em}
    \begin{equation} \label{equ:continuity}
    \begin{aligned}
        &\mathbb{E}\left[\big|\big|\alphaWeight{\Wb'} - \alphaWeight{\Wb}\big|\big|_2^2\right] \leq \epsilon', \; \forall m \in \mathbb{N} \; \text{and} \; \forall \mathbf{W}, \Wb' \in \mathbb{R}^{p \times p} \; \text{s.t.} \; ||\mathbf{W} - \Wb'||_2 \leq \epsilon.
    \end{aligned}
    \end{equation}
    \endgroup
    Since $|| \Wgeneric ||_2 \leq c$ for all $m \geq m'$, we have that $|| \alphaWeight{\Wgeneric} - \bm{\alpha} ||_2^2$ is also bounded by some constant $c' > 0$, for all $m \geq m'$, almost surely. We now fix an $\epsilon' > 0$ and choose a corresponding $\epsilon$ such that \eqref{equ:continuity} holds. We then conclude from \eqref{eq:first_summand_upper_bound} that

    \begingroup
    \addtolength{\jot}{.5em}
    \begin{equation*}
    \begin{aligned}
        \MSE \left( \alphaWeight{\West} \right) \geq \quad & \mathbb{E} \left[ \mathbbm{1}\left\{||\West - \Wgeneric||_2 \leq \epsilon \right\} ||\alphaWeight{\West} - \bm{\alpha}||_2^2 \right] \\
        \geq \quad & \MSE \left( \alphaWeight{\Wgeneric} \right) - 2 \sqrt{\epsilon' \; \mathbb{E}\left[ ||\alphaWeight{\Wgeneric} - \bm{\alpha}||_2^2 \right]} - P \left( || \West - \Wgeneric ||_2 > \epsilon \right) c' \\
        \geq \quad & \MSE \left( \alphaWeight{\Wgeneric} \right) - 2\sqrt{\epsilon' c'} - P \left( || \West - \Wgeneric ||_2 > \epsilon \right) c',
    \end{aligned}
    \end{equation*}
    \endgroup

    for all $m \geq m'$. Thus, we conclude 

    \begin{equation*}
        \ulim{m} \MSE \left( \alphaWeight{\West} \right) \; \geq \; \ulim{m} \; \MSE \left( \alphaWeight{\Wgeneric} \right) - 2\sqrt{\epsilon' c'}.
    \end{equation*}
    
    We can repeat this procedure for any $\epsilon' > 0$ and therefore conclude
    
    \begin{equation*}
        \ulim{m} \MSE \left( \alphaWeight{\West} \right) \; \geq \; \ulim{m} \; \MSE \left( \alphaWeight{\Wgeneric} \right),
    \end{equation*}
    which is the desired result.
\end{proof}

\begin{lemma2} \label{lem:upper_bounded}
Let $\West - \Wgeneric \asconv \mathbf{0}$ and let there exist some $c > 0$, $m' \in \mathbb{N}$, such that $||\Wgeneric||_2 \leq c, \forall m \geq m'$, almost surely. Then, it holds that
    \begin{equation*}
    \ulim{m} \; \MSE \, \big(\alphaWeight{\West}\big) \; \leq \; \ulim{m} \MSE \big(\alphaWeight{\Wgeneric} \big).
    \end{equation*}
\end{lemma2}
\begin{proof}
We again employ the formulation from \eqref{eq:upper_lower_bound}, but this time to construct an upper bound. For the first term of \eqref{eq:upper_lower_bound}, we see that
    \begingroup
    \addtolength{\jot}{.5em}
    \begin{equation} \label{eq:upper_term_2}
        \begin{aligned}
        & &&\mathbb{E}\left[ \mathbbm{1}\left\{||\West - \Wgeneric||_2 \leq \epsilon \right\} \; ||\alphaWeight{\West} - \bm{\alpha}||_2^2 \right] \; = \; \mathbb{E}\left[ \mathbbm{1}\left\{||\West - \Wgeneric||_2 \leq \epsilon \right\} \; ||\alphaWeight{\West} - \alphaWeight{\Wgeneric} + \alphaWeight{\Wgeneric} - \bm{\alpha}||_2^2 \right] \\
        &\leq \; &&\text{MSE} \, \big( \alphaWeight{\Wgeneric} \big) \; + \; 2\sqrt{\mathbb{E}\left[ \mathbbm{1}\left\{||\West - \Wgeneric||_2 \leq \epsilon \right\} ||\alphaWeight{\West} - \alphaWeight{\Wgeneric}||_2^2 \right] \; \mathbb{E}[||\alphaWeight{\Wgeneric} - \bm{\alpha} ||_2^2]}\; + \\ 
        & &&\mathbb{E}\left[ \mathbbm{1}\left\{||\West - \Wgeneric||_2 \leq \epsilon \right\} ||\alphaWeight{\West} - \alphaWeight{\Wgeneric}||_2^2 \right],
        \end{aligned}
    \end{equation}
    \endgroup
    by triangle inequality and the Cauchy-Schwarz inequality. Since for $m \geq m'$ it holds that $||\Wgeneric||_2 \leq c$, almost surely, there exists a constant $c' > 0$ such that $\mathbb{E}\left[||\alphaWeight{\Wgeneric} - \bm{\alpha}||_2^2\right] \; \leq \; c'$, for all $m \geq m'$. This is true because the two estimators $\alphaI$ and $\alphaO$ have both bounded mean squared error for any sample size $m$.

    Analogously to the proof for Lemma \ref{lem:MSEconvergence}, we now fix an $\epsilon' > 0$ and choose a corresponding $\epsilon$ such that \eqref{equ:continuity} holds. For $m \geq m'$, we then conclude from \eqref{eq:upper_term_2} that 
    \begingroup
    \addtolength{\jot}{.5em}
    \begin{equation}
    \begin{aligned} \label{eq:first_term}
        &\mathbb{E} \left[ \mathbbm{1}\left\{||\West - \Wgeneric||_2 \leq \epsilon \right\} ||\alphaWeight{\West} - \bm{\alpha}||_2^2 \right] \\
        \leq \quad & \MSE \left( \alphaWeight{\Wgeneric} \right) + 2 \sqrt{\epsilon' \; \mathbb{E}\left[ ||\alphaWeight{\Wgeneric} - \bm{\alpha}||_2^2 \right]} + \epsilon' \\
        \leq \quad & \MSE \left( \alphaWeight{\Wgeneric} \right) + 2\sqrt{\epsilon' c'} + \epsilon'.
    \end{aligned}
    \end{equation}
    \endgroup
    This bounds the first term of \eqref{eq:upper_lower_bound}. For the second term of \eqref{eq:upper_lower_bound}, we use almost sure convergence of $\West - \Wgeneric$. Since $\Wgeneric$ is bounded in the limit, almost surely, so is $\West$. Formally, $|| \West ||_2 \leq c^{''}, \forall m \geq m'$ for some $m' \in \mathbb{N}$, almost surely.
    
    We use this to bound $||\alphaWeight{\West} - \bm{\alpha}||_2^2 < c'''$ for all $m \geq m'$, almost surely, for some $c''' > 0$. Now, we apply iterated expectations to the second term of \eqref{eq:upper_lower_bound} to see that for all $m \geq m'$
    \begingroup
    \addtolength{\jot}{.5em}
    \begin{equation} \label{eq:second_term}
    \begin{aligned}
        \mathbb{E} \left[ \mathbbm{1}\left\{||\West - \Wgeneric||_2 > \epsilon \right\} ||\alphaWeight{\West} - \bm{\alpha}||_2^2 \right] \; 
        = \quad &\mathbb{E}_{\West} \left[ \mathbbm{1}\left\{||\West - \Wgeneric||_2 > \epsilon \right\} \; \mathbb{E}_{\alphaWeight{\West}|\West}\left[||\alphaWeight{\West} - \bm{\alpha}||_2^2 \right] \right] \\ 
        \leq \quad &\text{P} \left( ||\West - \Wgeneric||_2 > \epsilon \right) c''',
    \end{aligned}
    \end{equation}
    \endgroup
    almost surely. Now, we can combine the inequalities \eqref{eq:first_term} and \eqref{eq:second_term} to obtain
    \begin{equation*}
    \text{MSE} \, \left( \alphaWeight{\West} \right) \; \leq \; \text{MSE} \, \left( \alphaWeight{\Wgeneric} \right) \; + \; 2 \sqrt{\epsilon' c'} \; + \; \epsilon' \; + \; \text{P} \left( || \West - \Wgeneric ||_2 > \epsilon \right) c''',
    \end{equation*}
    for all $m \geq m''$. Almost sure convergence implies consistency of $\West - \Wgeneric$ with respect to $\mathbf{0}$, so we see that $\text{P}\left( || \West - \Wgeneric ||_2 > \epsilon \right)$ vanishes in the limit $m \rightarrow \infty$, for all $\epsilon > 0$. We can repeat this procedure for any $\epsilon' > 0$. This implies the desired result.
\end{proof}

\section{Detailed Derivation of Optimal Weighting Schemes} \label{sec:detailed_derivation}

In general, we observe that

\begingroup
    \addtolength{\jot}{.5em}
    \begin{equation*}
        \begin{aligned}
    \Bias(\alphaWeight{\Wb}) &= \Wb \alphab + (\Ib - \Wb)(\alphab + \bm{\Delta}) - \alphab = (\Ib - \Wb) \bm{\Delta}, \\
    \COV(\alphaWeight{\Wb}) &= \Wb \COV(\alphaI) \Wb^{\top} + (\Ib - \Wb) \COV(\alphaO) (\Ib - \Wb)^{\top}.
    \end{aligned}
\end{equation*}
\endgroup

\subsection{Optimal Scalar Weight}

Here, we have
\begin{align*}
&\frac{\partial}{\partial w} \MSE\left(\alphaWeight{w\Ib_p}\right) \\
= \quad &\frac{\partial}{\partial w} \Big| \Big| \Bias\left(\alphaWeight{w\Ib_p}\right) \Big| \Big|_2^2 \; + \; \frac{\partial}{\partial w}\text{Tr} \left( \COV \left(\alphaWeight{w\Ib_p}\right) \right) \\
= \quad & -2(1 - w) ||\bm{\Delta}||_2^2 + 2w \text{Tr}\left( \COV(\alphaI) \right) - 2(1 - w) \text{Tr}\left( \COV(\alphaO) \right) \overset{!}{=} 0.
\end{align*}

By rearranging, we get

\begin{equation*}
\wopt = \frac{\text{Tr}(\mathbf{Cov}(\alphaO)) + \norm{\bm{\Delta}}_2^2}{\text{Tr}(\mathbf{Cov}(\alphaI)) + \text{Tr}(\mathbf{Cov}(\alphaO)) + \norm{\bm{\Delta}}_2^2}.
\end{equation*}

\subsection{Optimal Diagonal Weight Matrix}

Here, we see that the objective decouples into a sum over the individual dimensions

\begin{equation*}
    \MSE\left(\alphaWeight{w\Ib_p}\right) = \sum_{k = 1}^p \left( 1 - w^{(k)} \right)^2 \bm{\Delta}^{(k) \, 2} \; + \; w^{(k) \, 2} \mathbf{Cov}^{(k, k)}(\alphaI) +  \left( 1 - w^{(k)} \right)^2 \mathbf{Cov}^{(k, k)}(\alphaO).
\end{equation*}

Thus, we optimize for each dimension $k$ separately and obtain

\begin{equation*}
    w_*^{m (k)} = \frac{\mathrm{Cov}^{(k, k)}(\alphaO) + \Delta^{(k) \, 2}}{\mathrm{Cov}^{(k, k)}(\alphaI) + \mathrm{Cov}^{(k, k)}(\alphaO) + \Delta^{(k) \, 2}}.
\end{equation*}

\subsection{Optimal Weight Matrix}

Using $\frac{\partial}{\partial \Wb} \text{Tr}(\Wb \Ab \Wb^{\top}) = 2 \Wb \Ab$, since $\Ab$ is symmetric, we observe that
\begin{align*}
    &\frac{\partial}{\partial \Wb} \MSE\left(\alphaWeight{\Wb}\right) \\
    = \quad &2\Wb \left( \COV(\alphaI) + \COV(\alphaO) + \bm{\Delta} \bm{\Delta}^{\top} \right) - 2 \left( \bm{\Delta} \bm{\Delta}^{\top} + \COV(\alphaO) \right) \\
    \overset{!}{=} \quad &\mathbf{0}.
\end{align*}
We see that this minimum is attained for
\begin{equation*}
\left( \COV(\alphaO) + \bm{\Delta} \bm{\Delta}^{\top} \right) \left( \COV(\alphaI) + \COV(\alphaO) + \bm{\Delta} \bm{\Delta}^{\top} \right)^{-1}.
\end{equation*}
\section{Non Zero-Mean Exogenous Variables} \label{sec:mean_non_zero}

All results established here can readily be extended to settings, where any of the exogenous variables have non-zero mean, i.e., $\bm{\mu}_{\Nb_\Xb}$, $\bm{\mu}_{\tilde{\Nb}_\Xb} \coloneqq \mathbb{E}[\tilde{\Nb}_\Xb]$, $\bm{\mu}_{\Nb_\Zb}$, $\mu_{N_Y}$ (see \eqref{equ:multivarSCMconfounder}--\eqref{eq:outcome}) may be non-zero. In order to extend the practical estimators introduced here, one needs to consider the following two pre-processing steps:

First, we center both treatment distributions separately, without scaling:

\begin{align}
    &\xb'_i \; \leftarrow \; \xb_i - n^{-1}\sum_{j \in 1, ..., n} \xb_j, \quad &&\forall i \in 1, ..., n ,\label{eq:mean_center} \\
    &\xb'_i \; \leftarrow \; \xb_i - m^{-1}\sum_{j \in n+1, ..., n+m} \xb_j, \quad &&\forall i \in n+1, ..., n+m.
\end{align}
In this manner, both treatment variables become zero-mean. 

Furthermore, we add a dummy dimension with value one to all treatment vectors:

\begin{equation*}
    \xb''_i \; \leftarrow \; (\xb'_i, \; 1), \quad \forall i \in 1, ..., n+m. 
\end{equation*}

This naturally adds one more dimension also to $\bm{\alpha}$, which corresponds to the intercept term. We then use the constructed $\xb''_i$ to compute the weight matrices proposed in this work.

Finally, we see that the intercept term must be identical for both distributions, interventional and observational:

\begin{equation*} 
    \mathbb{E}[Y \; | \; \Xb' = \xb'] \; = \; \gammab^{\top} \EE[\Zb \; | \; \Xb'= \xb'] + \alphab^\top \xb' + \mu_{N_Y}.
\end{equation*}

We then have in the observational setting (data points $1, ..., n$) that

\begin{align*}
    \bm{\gamma}^{\top} \EE[\Zb \; | \; \Xb'= \xb'] &= \bm{\gamma}^\top \bm{\mu}_{\Nb_\Zb} + \bm{\gamma}^\top \Sigmab_{\Nb_\mathbf{Z}} \mathbf{B}^{\top} (\Sigmab_{\mathbf{N}_{\Xb}} + \mathbf{B} \Sigmab_{\Nb_\mathbf{Z}} \mathbf{B}^{\top})^{-1} (\xb' - \mathbb{E}[\Xb']) \\
     &= \bm{\gamma}^{\top} \bm{\mu}_{\Nb_\Zb} + \bm{\Delta}^{\top} \xb',
\end{align*}

where $\mathbb{E}[\Xb'] = \mathbf{0}$ due to ~\eqref{eq:mean_center}.

For the interventional data, we have independence between $\Xb'$ and $\Zb$ by definition and so we trivially get

\begin{equation*}
    \bm{\gamma}^{\top} \EE[\Zb \; | \; \Xb'= \xb'] = \bm{\gamma}^{\top} \bm{\mu}_{\Nb_\Zb}
\end{equation*}

here. Thus, the intercept is $\bm{\gamma}^{\top} \bm{\mu}_{\Nb_\Zb} + \mu_{N_Y}$ for both distributions and we fix $\hat{\Delta}^{(p+1)} = 0$.

\section{Sample Imbalance} \label{sec:sample_imbalance}

We see that the ground truth covariance matrices of $\alphaI$ and $\alphaO$ adapt to changes in the sample sizes, keeping the distributions of all variables fixed. For instance, we see that

\begin{equation*}
    \COV(\alphaI) = (\XI^{\top} \XI)^{-1} \sigma_{Y|\text{do}(X)}^2 = m^{-1} (m^{-1} \XI^{\top} \XI)^{-1} \sigma_{Y|\text{do}(X)}^2.
\end{equation*}

The term $(m^{-1} \XI^{\top} \XI)^{-1} \sigma_{Y|\text{do}(X)}^2$ is bounded in probability, for large enough $m$. Accordingly, this implies that $\COV(\alphaI) \cons \mathbf{0}$. Thus, when keeping $n$ fixed, we obtain $\Wopt \cons \Ib_p$, for $m \rightarrow \infty$. 

On the other hand, if we keep $m$ fixed and consider the limit $n \rightarrow \infty$ instead, we observe that

\begin{equation*}
    \Wopt \cons \Deltab \Deltab^{\top} (\COV(\alphaI) + \Deltab \Deltab^{\top})^{-1}.
\end{equation*}

We note that we do not have $\Wopt \cons \mathbf{0}$ here in general, because the bias in $\alphaO$ remains, independent of the sample size $n$.

\end{document}